\documentclass[letterpaper,twocolumn,10pt]{article}
\usepackage{usenix}

\usepackage{amssymb}
\usepackage{adjustbox}
\usepackage[numbers]{natbib}
\usepackage[utf8]{inputenc} 
\usepackage[T1]{fontenc}    
\usepackage{hyperref}       
\usepackage{url}            
\usepackage{booktabs}       
\usepackage{amsfonts}       
\usepackage{nicefrac}       
\usepackage{microtype}      
\usepackage{lipsum}
\usepackage{fancyhdr}       
\usepackage{graphicx}       
\graphicspath{{media/}}     

\usepackage{tikz}

\usepackage{wrapfig}       
\usepackage{graphics}
\usepackage{circledsteps}

\newlength{\figwidth}
\usepackage{flushend}

\definecolor{mygreen}{rgb}{0,0.6,0}
\definecolor{mygray}{rgb}{0.1,0.1,0.1}
\definecolor{mymauve}{rgb}{0.58,0,0.82}

\usepackage{multicol}
\usepackage[
	operators,
	adversary
	]{cryptocode}

\usepackage{csquotes}
\usepackage{dashbox}

\usetikzlibrary{shapes.callouts}
\usepackage{listings}

\usepackage{trace}

\usepackage{makeidx}
\usepackage{mdframed}

\makeindex

\PassOptionsToPackage{gray}{xcolor}

\newcommand*\sqnew[1][0.8ex]{\tikz\draw[thick,draw=green!60!black,fill=green!60!black] (0,0) rectangle (0.175,0.175);}
\newcommand*\emptycirc[1][0.8ex]{\tikz\draw[thick,draw=black] (0,0) circle (#1);} 
\newcommand*\halfcirc[1][0.8ex]{%
  \begin{tikzpicture}
  \draw[fill=black!60!black,draw=black!60!black] (0,0)-- (90:#1) arc (90:270:#1) -- cycle ;
  \draw[thick,draw=black!60!black] (0,0) circle (#1);
  \end{tikzpicture}}
\newcommand*\fullcirc[1][0.8ex]{\tikz\draw[thick,draw=black!60!black,fill=black!60!black] (0,0) circle (#1);} 

\usetikzlibrary{matrix,shapes,arrows,positioning,chains, calc}

\pgfdeclareshape{slash underlined}
{
  \inheritsavedanchors[from=rectangle] 
  \inheritanchorborder[from=rectangle]
  \inheritanchor[from=rectangle]{north}
  \inheritanchor[from=rectangle]{north west}
  \inheritanchor[from=rectangle]{north east}
  \inheritanchor[from=rectangle]{center}
  \inheritanchor[from=rectangle]{west}
  \inheritanchor[from=rectangle]{east}
  \inheritanchor[from=rectangle]{mid}
  \inheritanchor[from=rectangle]{mid west}
  \inheritanchor[from=rectangle]{mid east}
  \inheritanchor[from=rectangle]{base}
  \inheritanchor[from=rectangle]{base west}
  \inheritanchor[from=rectangle]{base east}
  \inheritanchor[from=rectangle]{south}
  \inheritanchor[from=rectangle]{south west}
  \inheritanchor[from=rectangle]{south east}
  \inheritanchorborder[from=rectangle]
  \foregroundpath{
    \southwest \pgf@xa=\pgf@x \pgf@ya=\pgf@y
    \northeast \pgf@xb=\pgf@x \pgf@yb=\pgf@y
    \pgf@xc=\pgf@xa
    \advance\pgf@xc by .5\pgf@xb
    \pgf@yc=\pgf@ya
    \advance\pgf@xc by -1.3pt
    \advance\pgf@yc by -1.8pt
    \pgfpathmoveto{\pgfqpoint{\pgf@xc}{\pgf@yc}}
    \advance\pgf@xc by  2.6pt
    \advance\pgf@yc by  3.6pt
    \pgfpathlineto{\pgfqpoint{\pgf@xc}{\pgf@yc}}
    \pgfpathmoveto{\pgfqpoint{\pgf@xa}{\pgf@ya}}
    \pgfpathlineto{\pgfqpoint{\pgf@xb}{\pgf@ya}}
 }
}
\tikzset{nomorepostaction/.code=\let\tikz@postactions\pgfutil@empty}

\makeatother

\newcommand{\avecq}{\ensuremath{\mathsf{AVeCQ}}~}

\makeatletter
\def\slashedarrowfill@#1#2#3#4#5{%
  $\m@th\thickmuskip0mu\medmuskip\thickmuskip\thinmuskip\thickmuskip
  \relax#5#1\mkern-7mu%
  \cleaders\hbox{$#5\mkern-2mu#2\mkern-2mu$}\hfill
  \mathclap{#3}\mathclap{#2}%
  \cleaders\hbox{$#5\mkern-2mu#2\mkern-2mu$}\hfill
  \mkern-7mu#4$%
}

\def\rightslashedarrowfilla@{%
  \slashedarrowfill@\relbar\relbar{\raisebox{1.2pt}{$\scriptscriptstyle\diagup$}}\rightarrow}
\newcommand\xslashedrightarrowa[2][]{%
  \ext@arrow 0055{\rightslashedarrowfilla@}{#1}{#2}}

\def\rightslashedarrowfillb@{%
  \slashedarrowfill@\relbar\relbar/\rightarrow}
\newcommand\xslashedrightarrowb[2][]{%
  \ext@arrow 0055{\rightslashedarrowfillb@}{#1}{#2}}

\def\rightslashedarrowfillc@{%
  \slashedarrowfill@\relbar\relbar{\raisebox{.12em}{\tiny/}}\rightarrow}
\newcommand\xslashedrightarrowc[2][]{%
  \ext@arrow 0055{\rightslashedarrowfillc@}{#1}{#2}}

\usepackage{booktabs}
\usepackage{array}
\newcolumntype{x}[1]{>{\centering\arraybackslash\hspace{0pt}}p{#1}}

\usepackage{extarrows}

\usepackage{algorithm}
\usepackage[noend]{algpseudocode}
\usepackage{xcolor}
\usepackage{cleveref} 
\usepackage{float}

\usepackage{mathtools}

\usepackage{mathtools}

\usepackage[para]{footmisc}
\usepackage[abs]{overpic}
\usepackage{varwidth}

\usepackage{longtable}
\usepackage{threeparttable}  
\usepackage{float}

\usepackage{enumitem}
\usepackage{arydshln}
\usepackage{tcolorbox}
\usepackage{graphicx}
\usepackage{textcomp}
\usepackage{algorithm} 
\usepackage{multirow}
\usepackage{ dsfont }
\usepackage{multirow,multicol}

\def\BibTeX{{\rm B\kern-.05em{\sc i\kern-.025em b}\kern-.08em
    T\kern-.1667em\lower.7ex\hbox{E}\kern-.125emX}}

\usepackage{amsmath,amsfonts}
\usepackage{textcomp}
\usepackage{xcolor}
\usepackage{enumitem}
\def\BibTeX{{\rm B\kern-.05em{\sc i\kern-.025em b}\kern-.08em
    T\kern-.1667em\lower.7ex\hbox{E}\kern-.125emX}}

\usepackage{makecell, tabularx}
\usepackage{ragged2e} 
\newcolumntype{L}{>{\RaggedRight}X} 

\usepackage{etoolbox}
\usepackage{tkz-euclide}
\tikzset{%
    pics/sema/.style args={#1/#2/#3}{code={%
        \ifstrequal{#2}{0}{%
            \node[circle,minimum width=.5mm,draw,fill=#1] {};
        }{%
            \tkzDefPoint(0,0){O}
            \tkzDrawSector[R,fill=#1](O,.6mm)(90,90-#2)
            \tkzDrawSector[R,fill=#3](O,.6mm)(90-#2,90-360)
    }
    }},
}

\usepackage{amsmath,amsfonts}
\usepackage{textcomp}
\usepackage{comment}
\def\BibTeX{{\rm B\kern-.05em{\sc i\kern-.025em b}\kern-.08em
    T\kern-.1667em\lower.7ex\hbox{E}\kern-.125emX}}

\usepackage{arydshln}
\usepackage{amsthm}
\newtheorem{theorem}{Theorem}

\newtheorem{definition}{Definition}

\usepackage{booktabs}
\usepackage{pifont}

\usepackage{pifont}
\usepackage{stackengine}
\newcommand\xrowht[2][0]{\addstackgap[.5\dimexpr#2\relax]{\vphantom{#1}}}

\begin{document}

\title{$\mathsf{AVeCQ}$: Anonymous Verifiable Crowdsourcing with~Worker~Qualities}

\date{}

\author{\rm Sankarshan Damle\\
{IIIT Hyderabad} \\
{sankarshan.damle@research.iiit.ac.in} 
\and 
 {\rm Vlasis Koutsos, Dimitrios Papadopoulos} \\
{HKUST} \\ \{vkoutsos,dipapado\}@cse.ust.hk
\and 
{\rm Dimitris Chatzopoulos} \\
{Univeristy College of Dublin} \\
  {dimitris.chatzopoulos@ucd.ie}
\and 
{\rm Sujit Gujar} \\
{IIIT Hyderabad} \\
  {sujit.gujar@iiit.ac.in}
}

\maketitle

\begin{abstract}
  In crowdsourcing systems, requesters publish tasks, and interested workers provide answers to get rewards. 
    \emph{Worker anonymity} motivates participation since it protects their privacy. 
    \emph{Anonymity with unlinkability} is an enhanced version of anonymity because it makes it impossible to ``link'' workers across the tasks they participate in. 
    Another core feature of crowdsourcing systems is \emph{worker quality} which expresses a worker's trustworthiness and quantifies their historical performance. Notably, worker quality depends on the participation history, revealing information about it, while unlinkability aims to disassociate the workers' identities from their past activity. 
    In this work, we present $\mathsf{AVeCQ}$, the first crowdsourcing system that reconciles these properties, achieving enhanced anonymity and verifiable worker quality updates. 
    \avecq relies on a suite of cryptographic tools, such as zero-knowledge proofs, to \emph{(i)} guarantee workers' privacy, \emph{(ii)} prove the correctness of worker quality scores and task answers, and \emph{(iii)} commensurate payments. \avecq is developed modularly, where the requesters and workers communicate over a platform that supports pseudonymity, information logging, and payments. In order to compare \avecq with the state-of-the-art, we prototype it over Ethereum. \avecq outperforms the state-of-the-art in three popular crowdsourcing tasks (image annotation, average review, and Gallup polls). For instance, for an Average Review task with $5$ choices and $128$ participating workers \avecq is 40\% faster (including overhead to compute and verify the necessary proofs and blockchain transaction processing time) with the task's requester consuming 87\% fewer gas units. 
\end{abstract}


\section{Introduction}

    \emph{Crowdsourcing} is the process of gathering information regarding a \emph{task} (e.g., a query or project) by leveraging a set of agents who are incentivized to work on them within a specific time frame~\cite{howe2006rise}.
    A prominent example of crowdsourcing revolves around Human Intelligence Tasks (HITs), which can be used to enrich datasets designed for empowering machine learning models.
    Those who request crowdsourcing tasks can extract statistical data, form conclusions, and even monetize from any results based on the individually-provided answers~\cite{marcus2012counting}.
    
    Specifically, a popular use case is the calculation of the average over a set of values. Representative examples can be found in personal data analytics (e.g., average salary calculation), smart agriculture (average crop collection), smart grid (average daily energy consumption), and others~\cite{zheng2017truth,sun2018truth,welinder2010multidimensional,he2016ups}.
    Other motivating tasks revolve around calculating a set's $n$-most popular items. 
    E.g., for $n=1$ this encompasses image annotation~\cite{krivosheev2020detecting,shah2015double}, while for $n>1$ Gallup polls~\cite{perez2022dataset}. 
    
    More concretely, a \emph{requester} publishes a task seeking information and \emph{workers} provide their responses.
    The requester can define a \emph{task policy} specifying various task parameters (e.g., the task description, a final answer calculation mechanism, a minimum number of participating workers). 
    To \emph{incentivize} participation, requesters may compensate workers
    ~\cite{MCS,gt1,Location}.  
    E.g., the workers may be compensated based on a flat rate or even based on how ``close'' their responses are to the final answer of the task. 
    Such compensation mechanisms can be specified in the task policy and in fact, there exists a plethora of deployed crowdsourcing systems that operate in this paradigm (e.g., Amazon MTurk~\cite{mturk}, Microwork~\cite{microwork}, and QMarkets~\cite{qmarkets}).
    
    \begin{figure}[t]
        \centering
        \includegraphics[width=0.9\columnwidth]{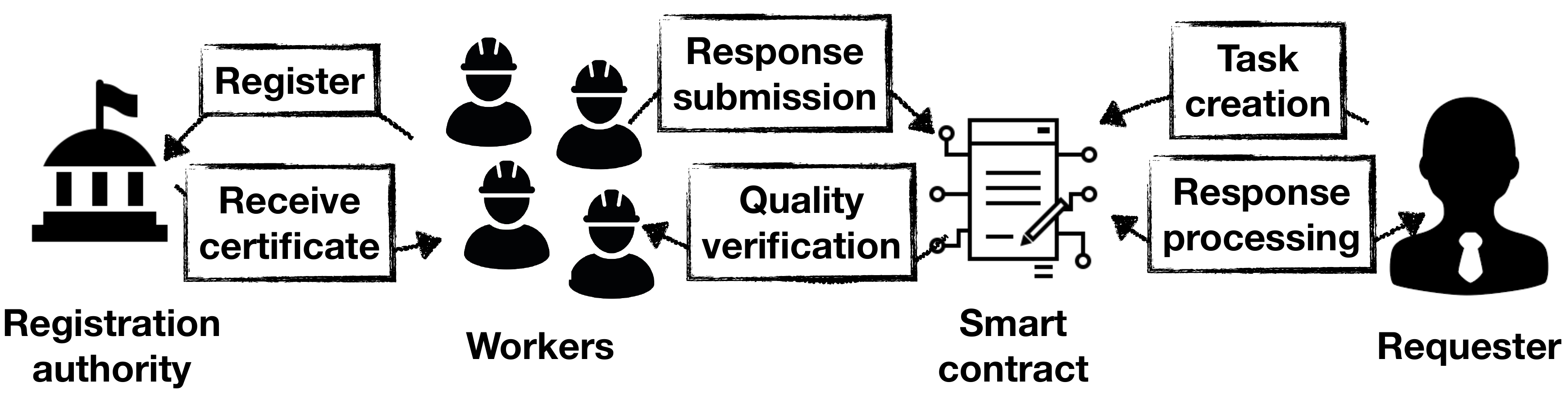}
     
        \caption{The crowdsourcing setting and phases of \avecq}
        \label{fig:overall}

    \end{figure}
    
    A common feature in crowdsourcing systems is that workers are individually associated with a \emph{quality score}, which estimates how ``trustworthy'' their responses are based on prior performance~\cite{peer2014reputation}.
    Indeed, worker qualities can be a vital tool for a requester, who can use them to screen workers (e.g., specifying a certain threshold so that only a worker whose quality surpasses it may participate), or even pay them according to their quality scores. 
    A worker's quality score (we refer to it as \emph{quality}) is dynamic, as it may change after every participation, depending on the task policy. 
    In practice, qualities represent the performance of workers in previous tasks. Thus, they can ultimately assist in ``ostracizing'' workers who submit answers without judiciously performing tasks~\cite{MCS,Gom1,Farm}. 
    E.g., workers who submit arbitrary answers will have their quality decrease over time, making it increasingly harder to clear quality thresholds and participate in future tasks. 
    
    \smallskip\noindent\textbf{Privacy in Crowdsourcing.} 
        Protecting worker's private information is greatly important~\cite{kandappu2015privacy}.
        In practice, task participation may require workers to disclose sensitive data to requesters (e.g., location, age, gender, race). Being able to observe the answer pattern of a specific worker is therefore undesirable; in fact, it opens the system up to worker profiling~\cite{kandappu2015privacy} and potential discrimination! For example, suppose Alice deploys a task requesting workers to disclose their racial background. After Bob provides such information, solely based on this, Alice may choose to exclude Bob from her future tasks.     
        
        Motivated by this, a line of works has emerged that studies privacy in crowdsourcing systems and proposes corresponding solutions~\cite{wang2016incentive,wang2018geographic,erlingsson2014rappor,wang2016differential,privacy2019ijcai,lu2018zebralancer,to2014framework,dynamo2015,li2018crowdbc}. 
        The required property in these works is \emph{worker anonymity}: it should be impossible to deduce a worker's identity from information revealed \emph{while participating in a task}. 
        A ``naive'' way to achieve anonymity would be to hide workers' identities behind pseudonyms.
        However, this is not enough as, through a series of tasks, requesters may still be able to identify workers on the basis of their answers alone.
        If a worker participates in multiple tasks, requesters may be able to build a ``rich'' profile \emph{linked} to a certain pseudonym.
        To avoid such a case, prior works~\cite{6888897,lu2018zebralancer} consider a stronger privacy notion, \emph{anonymity with unlinkability}: It should be impossible to link a worker's participation across tasks.
        Throughout the rest of the paper we refer to anonymity as in this ``stronger'' variant.
        
        We now make the following observations based on the discussion above, regarding qualities and anonymity. 
        On one hand, qualities directly stem from workers' participation profiles. 
        On the other hand, anonymity aims to obscure all past participation information. 
        Thus, the two properties appear to be \emph{inherently contradictory}: achieving one seemingly precludes the other. 
        In fact, there exist works that achieve anonymous crowdsourcing without worker quality~\cite{lu2018zebralancer,to2014framework,dynamo2015} and vice versa~\cite{li2018crowdbc} (see Table~\ref{tab:threeparttable}). To the best of our knowledge, no prior work simultaneously achieves both.
        
        \begin{table}[t]
            \centering
            \footnotesize
            \begin{threeparttable}
                \begin{tabular}{p{7em}p{2em}p{2em}p{2em}p{2em}p{2em}p{2em}}
                    & \rotatebox[origin=l]{45}{\raisebox{9pt}{\hspace{-0.5em}\multirow{2}{4.5em}{\textbf{Worker} \\ \textbf{Quality}}}} & \rotatebox[origin=l]{45}{\raisebox{8pt}{\hspace{-0.8em}\multirow{2}{4.5em}{\textbf{Anonymity}}}} & \rotatebox[origin=l]{45}{\raisebox{9pt}{\hspace{-0.5em}\multirow{2}{4.5em}{\textbf{Policy} \\ \textbf{Verifiable}}}} & \rotatebox[origin=l]{45}{\raisebox{9pt}{\hspace{-0.5em}\multirow{2}{4.5em}{\textbf{Data Con-} \\ \textbf{fidentiality}}}} & \rotatebox[origin=l]{45}{\raisebox{9pt}{\hspace{-0.5em}\multirow{2}{5em}{\textbf{Sybil} \\ \textbf{Resistance}}}}  & \rotatebox[origin=l]{45}{\raisebox{9pt}{\hspace{-0.5em}\multirow{2}{5em}{\textbf{Free-rider}\\ \textbf{Resistance}}}}  \\
                    \midrule
                    Yan \emph{et al.}~\cite{yan2021verifiable}  & \emptycirc & \emptycirc & $\fullcirc^{\boldsymbol{\dagger}}$ & \fullcirc &\fullcirc  & \fullcirc \\
                    {Duan \emph{et al.}~\cite{duan2019aggregating}} & \emptycirc & \emptycirc & $\fullcirc^{\boldsymbol{\ddagger}}$ & \fullcirc & \emptycirc  & \fullcirc \\
                    \hdashline
                    PACE~\cite{pace21}  & \emptycirc$^{\bigstar}$ & \emptycirc & $\fullcirc^{\boldsymbol{\dagger}} $ & \fullcirc & \emptycirc & $\fullcirc^{\boldsymbol{\dagger}}$\\ 
                    TrustWorker~\cite{gao2021trustworker}  & \emptycirc$^\bigstar$ & \emptycirc & $\fullcirc^{\boldsymbol{\dagger}}$ & \emptycirc & \fullcirc & \emptycirc \\
                    {Dragoon~\cite{lu2020dragoon}} & \emptycirc$^\bigstar$ & \emptycirc &  \fullcirc & \fullcirc  & \fullcirc & \fullcirc \\
                      {bHIT~\cite{bhitVLDB22}} & \emptycirc$^\bigstar$ & \emptycirc &  \fullcirc  & \fullcirc  & \fullcirc & \fullcirc \\
                    \hdashline
                    zkCrowd~\cite{zhu2019zkcrowd} & \emptycirc & \fullcirc & $\fullcirc^{\boldsymbol{\ddagger}/\blacklozenge}$ & \fullcirc  & \fullcirc  & \fullcirc \\ 
                    {ZebraLancer~\cite{lu2018zebralancer}}  &  \emptycirc & \fullcirc & \fullcirc  & \fullcirc  & \fullcirc &  \fullcirc \\
                    \hdashline
                    BPRF~\cite{jo2019bprf} &\fullcirc & \halfcirc & $\fullcirc^\blacklozenge$ & \emptycirc  & \fullcirc  & \fullcirc \\
                    \hline
                    \rule{0pt}{1\normalbaselineskip}$\mathsf{\textbf{AVeCQ}}$ & \fullcirc & \fullcirc & \fullcirc & \fullcirc & \fullcirc & \fullcirc\\
                    \bottomrule
                \end{tabular}
                
                \hspace{0.8em}\emptycirc/\fullcirc/\halfcirc: Absence/Presence/Weaker-version of the property,\\  
                \hspace{0.8em}$\bigstar$: Data quality,
                $\boldsymbol{\dagger}$: Semi-honest intermediary,
                $\boldsymbol{\ddagger}$: TEE,
                ${\blacklozenge} $: Blockchain
            
                \caption{\avecq vs related works. Special symbols denote weaker variants of the property are achieved.        \label{tab:threeparttable}}
            \end{threeparttable}
          
        \end{table}
        
    \smallskip\noindent\textbf{This Work.} 
        We propose {$\mathsf{AVeCQ}$} (Figure~\ref{fig:overall}), 
        a crowdsourcing system that is 
        \emph{the first to satisfy worker anonymity while maintaining worker qualities in a verifiable manner}.
        Table~\ref{tab:threeparttable} highlights the differences between \avecq and existing works in terms of achieved properties (see Section~\ref{sec:rw} for a more in-depth comparison).
        
        First, our system achieves anonymity with unlinkability, i.e., 
        requesters learn \emph{nothing} about participating workers except for their explicit task answers and possibly that their corresponding qualities are above the threshold specified in the task's policy (but, crucially, not the quality itself).
        Second, even though each worker's participation history remains hidden, \avecq supports \emph{verifiable} qualities, meaning that the following two conditions apply.
        To participate in a task, workers must prove they are using their \emph{correctly calculated} quality as \emph{derived} by their entire participation history. 
        Likewise, upon task completion, requesters must prove the correctness of the participants' quality updates according to the answers and the task policy.
        Crucially, the workers in \avecq verify their updated qualities without any information about the final answer except for what can be trivially inferred from the policy. Note that achieving each of this properties on its own is rather straight-forward (e.g., if we do not care about protecting worker identities it is easy to check that the correct quality score is used for each task); the challenge arises when simultaneously trying to achieve both. 
        
        \avecq is additionally secure against other significant threats. 
        Anonymity might allow workers to generate multiple identities arbitrarily (i.e., perform a Sybil attack) and reap extra payments~\cite{wangKDD20}. 
        \avecq avoids Sybil attacks by requiring workers certificates to be issued by a \emph{Registration Authority} (RA) before task participation. 
        Moreover, we counter \emph{free-rider attacks}, i.e., workers cannot submit someone else's response or use another's quality as their own to successfully participate in a task.
        Last, \avecq guarantees fair worker compensation according to task policy, i.e., the requester cannot ``cheat'' to avoid payments.
        
        \smallskip\noindent\textbf{{Overview of Challenges and Techniques.}} 
        At the core of our solution lies a wide range of cryptographic techniques, including zero-knowledge proofs (ZKPs)
        that allow a prover to convince a verifier about the correctness of a computation, without revealing any private information. 
        Below we state briefly the major challenges we encountered in this work and how we overcame them.
        
    \smallskip
    \noindent\underline{\textit{Privately updating qualities, verifiably.}} 
        To hide worker qualities from the task requester we have workers submit their qualities as additively homomorphic commitments. 
        This enables requesters to increment/decrement all workers' qualities without ever accessing any raw underlying quality. This calculation is based on the workers' answers, the final answer, and the task policy.
        Now, all workers know their own answer to the task and the policy, so when they see their updated quality it is trivial to check whether the change is computed correctly or not, \emph{if they also know the final answer}.  
        However, since \avecq does not reveal the final answer to the workers, the question of how to verify that the corresponding updates are honestly computed arises.   
        To this end, the requester is obligated to present individual ZKPs to all workers, regarding the correctness of their quality update.

    \smallskip
    \noindent\underline{\textit{Proof of quality-freshness.}} 
        If a worker's quality decreases after participating in a task she may be incentivised to discard the latest quality update and reuse her earlier quality in future tasks.
        Thus, we face the following problem: ``\emph{How can we ensure that the quality used is always the latest, according to all previous task participations?}''. To address this, we borrow and adapt a technique  used in privacy-preserving cryptocurrencies~\cite{hopwood2016zcash,kosba2016hawk}. 
        The proof that requesters compute for each worker after task completion  must also 
        pertain to the fact that the updated quality commitment of a worker has been appended as a leaf to a Merkle Tree which contains all quality commitments across all tasks. 
        A worker that wishes to participate in a task provides a ZKP that pertains to the fact that the re-randomized quality commitment she provides corresponds to the quality committed in one of the Merkle tree leaves, without revealing to which one. 
        In this manner, a worker that tries to benefit by reusing an earlier quality will break her anonymity due to the way this proof is crafted in $\mathsf{AVeCQ}$---specifically as the same Merkle leaf will need to be used again (see Section~\ref{sec:avecq}).
        Crucially, this is also easy to detect even long after the worker's participation.
        
    \smallskip
    \noindent\underline{\textit{Anonymous task participation.}}
        To anonymously participate in future tasks, the worker cannot utilize the quality commitment in the way it exists in the Merkle tree as that would trivially break anonymity.  
        To surpass this obstacle, a worker can re-randomize her quality commitment and provide the requester with a ZKP that corresponds to its correctly updated quality and is above the task participation threshold.
        However, a new question arises now: ``\emph{how can a requester be certain that the worker has not re-randomized an outdated quality commitment?}'' 
        To prevent such behavior, each worker must submit a cryptographic hash, including the quality commitment used in the ZKP above, concatenated with the unique identifier it provided to the RA during registration. 
        Thus, any worker trying to re-randomize/use a previous quality could be trivially detected.
        
        Nevertheless, if the hash is computed ``naively'' now the RA can launch possible de-anonymization attacks.
        To avoid this, essentially, the worker must not hash any public information (i.e., its latest quality commitment transmitted by the latest task requester) together with its identifier. 
        Instead, the requester, upon concluding its task, will transmit the updated quality of the worker combined with a ``dummy'' commitment, whose randmoness will communicate to the worker in a confidential manner.
        The worker can then decompose the commitment, extract its \emph{true} quality commitment, and continue participating in tasks without compromising its identity.
        
    \smallskip
    \noindent\underline{\textit{Resistance to free-riding.}} 
        Workers might attempt to participate effortlessly by (re)submitting someone else's response or quality. \avecq resists such free-riding attacks by tying each worker's answer to her quality and a unique ``payment wallet'' via a ZKP. Thus, no worker attempting to ``hijack'' someone else's response is able to produce valid proof for participation and get compensated.

    \smallskip
    \noindent\textbf{Implementing $\mathsf{\textbf{AVeCQ}}$.} 
        We implement a prototype of \avecq, whose smart contract component is deployed over Ethereum testnets, i.e., Rinkeby~\cite{rinkeby} and Goerli~\cite{goerli}. 
        To demonstrate the practicality and scalability of $\mathsf{AVeCQ}$, we report extensively on its performance focusing on computational and blockchain overheads, communication bandwidth, and monetary costs. 
        We test \avecq on three popular, real-world-inspired tasks~\cite{krivosheev2020detecting,zheng2017truth,sun2018truth}, i.e., image annotation, average review estimation, and Gallup polls, using the real datasets Duck~\cite{welinder2010multidimensional}, Amazon Review~\cite{he2016ups}, and COVID-19 Survey~\cite{perez2022dataset}.  
        Our  results show, somewhat surprisingly, that \avecq outperforms other state-of-the-art systems, even though it achieves a combination of stronger security properties and/or operates in a stronger security model (Table~\ref{tab:threeparttable}). 
        We benchmark $\mathsf{AVeCQ}$'s performance for representative tasks against the ``best'' systems with available  implementation. 
        For instance, for binary image annotation with $39$ workers our end-to-end (E2E) time is $<8$ minutes, vs. $<2$ hours for only $11$ workers in~\cite{lu2018zebralancer}.
        Note that contrary to ZebraLancer, $\mathsf{AVeCQ}$'s smart contract is solely used for storage purposes (e.g., no on-chain verfication happens---see Section~\ref{sec:avecq}). 
        \avecq also retains its edge in terms of gas consumption.
        E.g., for generating an average review with $128$ workers, a requester in \avecq consumes $\approx 4.3$M gas units, whereas in~\cite{duan2019aggregating} the requester requires $\approx 35$M gas units for just $100$ workers. 
        Last, \avecq consumes $<25\%$ of the gas required in~\cite{lu2020dragoon} (See Section~\ref{sec:exp} for a more detailed comparison).

    \smallskip\noindent\textbf{Our Contributions.}
        In summary, we construct a crowdsourcing system that bridges the gap between anonymity and worker qualities while being able to scale to real-world inspired task instances.
        We highlight the main contributions of our work as follows:
        
        \begin{enumerate}[leftmargin=*, topsep=0.1pt, itemsep=0pt]
            
            \item We design $\mathsf{AVeCQ}$, the first crowdsourcing system that guarantees the anonymity of participating workers across tasks while maintaining a quality system verifiably (see Section~\ref{sec:avecq}).
            Crucially, it does so without compromising functionality, as it can support arbitrary policies and tasks.
            
            \smallskip\item We provide definitions for three critical security and privacy properties of crowdsourcing systems: anonymity, free-rider resistance, and policy verifiability. 
            We prove that \avecq satisfies all three properties under standard assumptions (see Section~\ref{sec:secanal}). Additionally, in Appendix~\ref{subsec:sec_B}, we additionally show \avecq to be secure against other, various, popular attacks. 
            
            \smallskip\item We develop a prototype implementation of \avecq and test its performance thoroughly. 
            In terms of efficiency and scalability, our implementation is comparable, when not better, than other state-of-the-art systems providing less functionality (e.g., support only gold-standard tasks) or operate in a weaker threat model (see Section~\ref{sec:exp}).
            In fact, we provide an in-depth comparison with prior works, both qualitative (in terms of properties) and quantitative (in terms of three specific real-world tasks).
            
        \end{enumerate}
        


\section{Related Work}\label{sec:rw}

    \noindent\textbf{Worker Quality.}  %
        Prior works in crowdsourcing systems adopt different notions of ``quality''. The authors of~\cite{lu2020dragoon,pace21,gao2021trustworker,bhitVLDB22} interpret worker quality in terms of proximity to an ``estimated'' or ``final'' answer. Specifically, Lu {et al.}~\cite{lu2020dragoon} use a set of gold-standard tasks, whose final answer is known a priori to the requester and a posteriori to the workers, to determine the quality of the answers. 
        Despite being commonly used, this approach is rather limiting since it only works when the answer is known. It does not work for other popular crowdsourcing tasks e.g., Gallup polls. 
        The authors of~\cite{pace21,gao2021trustworker} assign scores to workers based on the proximity of their response to the mean of the submitted data.
        Contrary, (as also in~\cite{li2018crowdbc,jo2019bprf,MCS,wang2017melody}) we interpret worker quality as a representation of workers' entire historical task performance. This is not only more realistic but also strictly more general, since \avecq can capture all previous quality notions through different task policies. 
        
    \smallskip
    \noindent\textbf{Privacy.} 
        Another point of contention in the literature revolves around the definition of identity and data privacy---both of which are essential. 
        Anonymity with unlinkability protects the workers identities across tasks entirely, as explained before. 
        Contrary, data privacy limits what the system entities (e.g., blockchain nodes) learn about the workers' data.
        
        Anonymity requires two conditions to be met: \emph{(i)} the identity of a worker remains hidden and \emph{(ii)} the quality itself cannot be used to de-anonymize workers. 
        The authors of~\cite{6888897,7463023} satisfy \emph{(i)} by enabling workers to generate identities freely.
        However, this in turn allows workers to carry out Sybil attacks. 
        Similarly, CrowdBC~\cite{li2018crowdbc} suffers from the same limitation, with two additional drawbacks. Crucially,~CrowdBC stores qualities on-chain, in the plain. This poses a potential and significant risk to de-anonymizing workers, as explained before, rendering~\cite{li2018crowdbc} not anonymous. 
        
        Regarding \emph{data privacy} towards third parties (most commonly referred to as \emph{confidentiality}), it is commonplace to require workers to submit their responses in an encrypted manner as done in~\cite{wang2016incentive,bindschaedler2017plausible,li2018crowdbc,liu2019iiot,lu2018zebralancer,hu2018reputation,duan2019aggregating,pace21,lu2020dragoon}, with two exceptions. 
        First, in~\cite{gao2021trustworker} a deterministic encryption scheme is used meaning the Computing Server (the intermediary collecting worker responses) can access them in the plain, while in~\cite{jo2019bprf} all responses are already communicated in the plain. 
        The authors of~\cite{to2014framework,erlingsson2014rappor,to2015privgeocrowd,wang2017location,wang2018geographic,Luo2018differential,privacy2019ijcai,tao2020differentially} use differential privacy to protect workers' inputs. 
        However, noisy methods affect the correctness of the crowdsourcing process as they dilute the final answer. Thus, they are impractical when including qualities in the crowdsourcing model, since most quality update mechanisms are based on correlations between a worker's answer and the ``final answer'' of a task, and even more so when the set of possible answers is of limited size.
        
    \smallskip
    \noindent\textbf{Verifiable Policy.}
        To achieve policy verifiability, works such as~\cite{pace21,gao2021trustworker,yan2021verifiable,duan2019aggregating,zhu2019zkcrowd,jo2019bprf} either entrust \emph{(i)} a semi-honest intermediary, \emph{(ii)} a Trusted Execution Environment (TEE), e.g., Intel SGX, or \emph{(iii)} blockchain miners to carry out policy-related computations (e.g., calculate the final answer or rewards).  
        However, having to ``blindly'' trust intermediaries is not ideal, secure hardware is susceptible to side-channel attacks~\cite{sgx1,sgx2}, and on-chain computations jeopardise data confidentiality.
        \avecq is free of any such assumptions and is policy-verifiable under only cryptographic assumptions.

    \smallskip
    \noindent\textbf{Other Security Issues.} 
        Two common security threats in crowdsourcing are \emph{Sybil} and \emph{Free-riding} attacks.
        A countermeasure to Sybil attacks is employing a trusted RA that registers workers by issuing a certificate based on the worker's unique identifier (e.g., ID documentation)~\cite{DBLP:conf/esorics/TanasDH15,lu2018zebralancer,liu2019iiot,jo2019bprf,zhu2019zkcrowd}.  
        In fact, the authors of~\cite{li2018crowdbc,duan2019aggregating,gao2021trustworker} do not utilize an RA and fail to prevent such attacks. 
        Alternatively, in~\cite{duan2019aggregating} the authors argue that their \emph{incentive-compatible} mechanism disincentives worker misbehavior ---a strictly stronger assumption.
    
        To safeguard against free-riding, requesters can employ a trusted third party~\cite{7037593}, or couple workers' certificates with their public addresses~\cite{lu2018zebralancer}.
        CrowdBC~\cite{li2018crowdbc}, alternatively, requires workers to deposit funds to a smart contract to be eligible for a task. Workers are incentivized to exert effort, or risk their deposits. Instead, \avecq eliminates free-riding attacks also by relying only on cryptographic assumptions.
        
    \smallskip
    \noindent\textbf{\avecq vs. State-of-the-art.}
        Unlike the works above, our system satisfies all properties in Table~\ref{tab:threeparttable}.
        Recently, Liang et al. proposed bHIT~\cite{bhitVLDB22}, a blockchain-based crowdsourcing system for HITs that, similarly to $\mathsf{AVeCQ}$, does not disclose the responses of the workers to the public blockchain.
        However, bHIT does not provide any notion of anonymity or policy verifiability and operates in a weaker security model since it does not consider colluding workers.
        
        Zebralancer~\cite{lu2018zebralancer} is the closest work to ours in terms of the employed techniques and properties. 
        Both \avecq and Zebralancer utilize zk-SNARKs but Zebralancer only uses them for proving the correctness of the rewards calculation, while in \avecq zk-SNARKs are also used by workers to prove that they are using their latest valid qualities. Moreover, requesters use zk-SNARKs to prove the correctness of the final answer, updates of qualities, and calculation of payments. Similarly, both works employ an RA, and utilize smart contracts for task deployment and participation. 
        However, in contrast to Zebralancer, \avecq additionally supports worker qualities. 
        This extra feature is far from trivial to implement as new security and privacy concerns emerge resulting in a more complex protocol.
        Nevertheless, this does not come at a performance cost, as \avecq is at least as (or even more) efficient than Zebralancer, despite of supporting worker qualities.

%
%

\section{Preliminaries\label{sec:prelim}}
    
    We now present the tools used in $\mathsf{AVeCQ}$. 
    Table~\ref{tab:notation} provides a reference for key notation.
    Let $\mathbb{E}$ be an elliptic curve defined over a large prime field $\mathbb{F}_p$ with $G,H\in\mathbb{E}$ as publicly known generators. 
    We denote by $x\sample A$ the sampling at random of the element $x$ from the domain $A$. 
    We denote by $\lambda$ a security parameter and by $negl(\lambda)$ a function negligible in $\lambda$.
    Last, we denote by $Adv^\mathcal{G}(\adv)$ the advantage that $\adv$ has in winning the $\mathcal{G}$ game.
    
    \noindent\textbf{Pedersen Commitments~\cite{pedersen1991non}.}    
        A commitment scheme binds and hides a value $x$.   
        Specifically, a Pedersen commitment of $x$ with randomness $r$ is in the form of $\textsf{Com}(x,r)=x\cdot G + r\cdot H$.
        Pedersen commitments are \emph{additively homomorphic}, i.e., $\textsf{Com}(x_1,r_1)\circ \textsf{Com}(x_2,r_2)=\textsf{Com}(x_1+x_2,r_1+r_2)$, \emph{computationally binding} (it is not feasible to ``change one's mind'' after committing), and  \emph{perfectly hiding} (they reveal nothing about the committed data).

    \begin{table}[t]
        \footnotesize
        \begin{tabular}{p{0.3\columnwidth}p{0.61\columnwidth}}
            \toprule
            \multicolumn{2}{c}{\textbf{Crowdsourcing Model}} \\
            \midrule
            $\mathcal{A}^{\tau}$    & Set of possible answers for task $\tau$\\
            $a_\phi^\tau$,$P^{\tau}$  & Final answer and task policy for  $\tau$ \\
            $d^\tau$, $\mathcal{DL}^\tau$ & Description and deadlines  of  $\tau$ \\
            $q_i^\tau$ & Quality score of $w_i$ at  $\tau$ \\
            $n_{th}^\tau$ & Threshold number of workers required for  $\tau$ \\
            $n^\tau$ & Total number of worker responses for  $\tau$ \\
            $p_i^\tau$ & $w_i$'s payment for $\tau$ \\
            $pa_i^\tau$ & $w_i$'s address for compensation for $\tau$ \\
            \midrule
            \multicolumn{2}{c}{\textbf{Protocol Notations}} \\
            \midrule
            $cert_i$ & EdDSA signature for party $i$ \\
            $root_{MT}$  & Root of a Merkle Tree $(MT)$ \\
            $path_{leaf}$ &  path for $leaf$ in $MT$ \\
            $r_{k,i}^\tau $ & $w_i$'s randomness for  $\tau$ \\
            $r_{*,i}^\tau $ & $w_i$'s randomness to re-randomize $\textsf{Com}(q_i^{\tau-1},\cdot)$ \\
            $r_{**,i}^\tau $ & Randomness of the dummy commitment for $q_i^\tau$ \\
            $r_{c,i}^\tau $ & Randomness with $\textsf{Com}(0,r_{**,i}^\tau)$ for $q_i^\tau$  \\
            $r_{c-d,i}^\tau $ & Randomness without $\textsf{Com}(0,r_{**,i}^\tau)$ for $q_i^\tau$ \\
            $E(pk_R,a_i^\tau;\cdot)$ & Encryption of $w_i$'s response $a_i^\tau$ for  $\tau$ \\
            $E(pk_R,pa_i^\tau;\cdot)$ & Encryption of $w_i$'s address $pa_i^\tau$ for  $\tau$ \\
            $\textsf{Com}(q_i^{\tau-1},r_{c,i}^{\tau-1}+r_{\star,i})$ & Re-randomized commitment of $q_i$ for  $\tau$ \\
            $E(pk_R,r_{k,i}^\tau;\cdot)$ & Encryption of $w_i$'s index $r_{k,i}^\tau$ for  $\tau$ \\
            $H\left(\textsf{Com}(q_i^{\tau-1},r_{c,i}^{\tau-1}),m_i\right)$ & Quality ``tag'' of $w_i$ for  $\tau$ \\
            \midrule
            \multicolumn{2}{c}{\textbf{zk-SNARKs}} \\
            \midrule
            $\pi_{o_i}^\tau$  & \textsc{ProveQual} proof for   $w_i$'s tuple $o_i^\tau$ for  $\tau$ \\
            $\pi_{a_\phi}^\tau$  & \textsc{AuthCalc} proof for $a_\phi^\tau$ for  $\tau$ \\
            $\pi_{q_i}^\tau$   & \textsc{AuthQual} proof for $w_i$'s updated $q_i^\tau$ for  $\tau$ \\
            $\pi_{a_i}^\tau$   & \textsc{AuthValue} proof for   $w_i$'s response $a_i^\tau$ for  $\tau$ \\
            \bottomrule
        \end{tabular}
    
        \caption{\label{tab:notation}Key Notations}
      
    \end{table}

    \smallskip
    \noindent\textbf{Public-Key Encryption (PKE) Scheme~\cite{koblitz1987elliptic}.}  
        A PKE scheme consists of the following algorithms:
        
        \begin{itemize}[leftmargin=*, topsep=0.1pt, itemsep=0pt]
            
            \item\textsf{KeyGen}$(\lambda)\rightarrow (sk,pk)$. Given the security parameter $\lambda$, \textsf{KeyGen} samples a secret key $sk\sample\{0,1\}^\lambda$, computes the public key $pk=sk\cdot G$. It outputs the key-pair ($sk,pk$).
            
            \item\textsf{Enc}$(pk,x;r)\rightarrow E(pk,x;r)$. To encrypt a value $x$, the algorithm takes input a randomness $r$ and outputs the curve point $(r\cdot G,P_x+\cdot r(sk\cdot G))$. Here, $P_x$ is a publicly-known mapping of a value $x$ to a curve point in $\mathbb{E}$. 
            
            \item\textsf{Dec}$(E(pk,x;r),sk)\rightarrow x$. To decrypt $x$ from $E(pk,x;r)$, the algorithms computes $x:=P_x+\cdot r(sk\cdot G) - r\cdot sk\cdot G$.
        \end{itemize}
    
    \smallskip
    \noindent\textbf{Hash Function~\cite{damgaard1989design}.} 
        A cryptographic hash function $H:\{0,1\}^*\rightarrow \{0,1\}^\lambda$ is collision-resistant if the probability of two distinct inputs mapping to the same output is negligible: $\Pr[H(x)=H(y) \mid x\not=y] \leq negl(\lambda)$.
        Additionally, it is pre-image resistant if the probability of inverting it is negligible. We denote the indistinguishability games for these properties as  $H-CR$ and $H-PR$ respectively.

    \smallskip
    \noindent\textbf{Digital Signatures~\cite{ches-2011-24091}.} 
        A digital signature scheme allows verification of the authenticity of a certificate. 
        EdDSA is a  Schnorr-based signature scheme defined over $\mathbb{E}$. 
        In EdDSA, given $G$, one derives its public key $pk$ by sampling $sk\sample \mathbb{F}_p$. 
        A party $i$ signs the value $H(m_i)$ for a secret message $m_i$ denoted as its signature $cert_i=(R_i,S_i)$. 
        Here, $R_i=r\cdot G$ s.t. $r\sample\mathbb{F}_p$, and $S_i=r+H(m_i)\cdot sk$. 
        A verifier accepts the signature iff $S_i\cdot G = R_i + H(m_i)\cdot pk$ holds. 
        
    \smallskip
    \noindent\textbf{zk-SNARKs~\cite{sasson2014zerocash}.} 
        A zero-knowledge succinct non-interactive argument of knowledge (zk-SNARK) allows a prover to convince a verifier about the correctness of a computation on private input, through a protocol. 
        The verifier-available information is referred to as \emph{statement} $\Vec{x}$ and the private input of the prover as \emph{witness} $\Vec{\omega}$. 
        The protocol execution takes place in a non-interactive manner, with succinct communication.
        A zk-SNARK consists of:
        \emph{(i)} a \emph{Setup} algorithm which outputs the public parameters $PP$ for a NP-complete language $\mathcal{L}_\mathcal{S}=\{\Vec{x}~|~\exists~\Vec{\omega} \mbox{~s.t.~}\mathcal{S}(\Vec{x},\Vec{\omega})=1\}$, where $\mathcal{S}:\mathbb{F}^n\times \mathbb{F}^h\rightarrow\mathbb{F}^l$ is the \emph{arithmetic circuit satisfiability problem} of an $\mathbb{F}$-arithmetic circuit; 
        \emph{(ii)} A \emph{Prover} algorithm that outputs a constant size proof $\pi$, attesting to the correctness of $\Vec{x}\in\mathcal{L}_S$ with witness $\Vec{\omega}$; 
        and \emph{(iii)} A \emph{Verifier} algorithm which efficiently checks the proof. 
        Informally, a zk-SNARK satisfies the following properties: 
        \begin{itemize}[leftmargin=*, topsep=0.1pt, itemsep=0pt]
            \item \textit{Completeness.} If $\exists\Vec{\omega}: \mathcal{L}_\mathcal{S}(\Vec{x},\Vec{\omega})=1$, honest Provers can always convince  Verifiers.
          
            \item \textit{Soundness.} If $\exists\Vec{\omega}^\prime: \mathcal{L}_\mathcal{S}(\Vec{x},\Vec{\omega})=0$, a dishonest Prover has negligible probability in convincing the Verifier.
          
            \item \textit{Zero-knowledge.} If $\exists \Vec{\omega}: \mathcal{L}_\mathcal{S}(\Vec{x},\Vec{\omega})=1$, the Verifier does not learn any information about $\Vec{\omega}$ (besides its existence).
        \end{itemize}

    \smallskip
    \noindent\textbf{Merkle Tree (MT)~\cite{merkle1987digital}.} 
    A Merkle tree ($MT$) is a complete binary tree where each parent node is a hash of its children. 
    This structure allows for \emph{membership} proofs, attesting to the existence of a specific leaf via a publicly known root ($root_{MT}$) and a path ($path_{leaf}$).

    \smallskip
    \noindent\textbf{Blockchain \& Smart Contracts.} 
        A blockchain is a ledger distributed across peers and made secure through cryptography and incentives. 
        Peers agree upon storing information in the form of blocks through consensus algorithms.
        Blockchain technology has transcended its use in cryptocurrency applications, especially with the introduction of smart contracts with Ethereum~\cite{wood2014ethereum}. 
        A smart contract is a computer program that can be run in an on-chain manner.
        Performing any smart-contract computation  over Ethereum requires \emph{gas}. 
        The amount of gas charged depends on the type of computation. More computationally extensive operations require higher gas to be executed on-chain.
        The total charge for the computation is referred to as \textit{gas cost}. 
        Any computation that alters the {state} of the contract, i.e., alters any contract method or variables, consumes gas.
        Contrary, reading data from the contract is free. 
        To submit state-altering transactions users specify a \emph{gas price} that they are willing to pay per \emph{gas unit}.

    \smallskip
    \noindent\textbf{EIP-1559~\cite{buterin2019eip}.} 
        In Ethereum, each transaction creator pays a dynamic base fee $b$ and a priority fee $\delta$ (in gas). Each block's miner receives $\delta$ while the base fee is ``burned'' (i.e., removed from the supply, forever). $\delta$ affects the verification time, as higher $\delta$ results in faster transaction verification time.

    \smallskip
    \noindent\textbf{Crowdsourcing Quality-Update Policies.} 
        Various techniques have been proposed for updating workers' qualities after task participation~\cite{10.1145/1401890.1401965, karger2011iterative,8024034,7065282,buchegger2003robust}. 
        \avecq integrates the widely adopted in academia~\cite{8024034,7065282} and industry (Amazon MTurk~\cite{mturk}) \emph{Beta} distribution.
        Formally, for each $w_i$ and task $\tau$, we maintain integers $\alpha_i^\tau$ and $\beta_i^\tau$, initialised to $1$. 
        The quality score for $\tau$ is then given by the Beta distribution with mean $\frac{\alpha_i^\tau}{\alpha_i^\tau+\beta_i^\tau}$.
        The update rule considering a worker's response $a_i^\tau$ and the ``final answer'' $a_\phi^\tau$ is: 
        \begin{itemize}[leftmargin=*, topsep=0.1pt, itemsep=0pt]
            \item if $a_i^\tau = a_\phi^\tau: \alpha_i^{\tau+1} = \alpha_i^{\tau} + 1; ~\beta_i^{\tau+1} = \beta_i^\tau$
    
            \item if $a_i^\tau \neq a_\phi^\tau: \beta^{\tau+1} = \beta^{\tau} + 1; ~\alpha_i^{\tau+1} = \alpha_i^{\tau}$
    
        \end{itemize}
    
        \noindent \avecq uses this mechanism to handle quality updates and specifically computes increments/decrements as Pedersen commitments of $1$ and $0$, accordingly. That way a worker can utilize the commitments' additive homomorphic property to compute its new quality.

%
%

\section{Problem Formulation}\label{sec:model}
    In this section, we present the problem formulation in three aspects: system model, threat model, and problem statement.

    \smallskip
    \noindent\textbf{System Model.}
        It includes three types of entities, namely a \emph{Registration Authority} (RA), \emph{a set of Requesters}, and \emph{a set of Workers}. 
        The RA is in charge of registering workers in the system by providing them with a participation certificate, upon receiving a unique identifier.
        Requesters can \emph{(i)} create and publish tasks, and \emph{(ii)} collect worker-generated responses.
        Workers observe published tasks and upon wishing to participate in a task they provide the respective requester with their individual responses.
        At a high level, to create a task, a requester needs to disclose its description alongside an answer-calculation mechanism, a quality-update rule, and a payment scheme. 
        Remember that each worker has a \emph{quality} that reflects their trustworthiness based on their previous answers. 
        Specifically, we denote all registered workers by the set $\mathcal{W} = \{w_{1},\ldots,w_{n}\}$, and their associated qualities by the set $\mathcal{Q} = \{q_{1},\ldots,q_{n}\}$.

    \smallskip
    \noindent\textbf{Threat Model.}
        We make no assumptions as to the behavior of requesters or workers. 
        Malicious requesters (R) can try to infer information about participating workers (not trivially leaked by their answers), calculate the final answer and the updated qualities arbitrarily, and avoid payments.
        On the other hand, malicious workers (W) can try to generate multiple identities arbitrarily and respond by utilizing outdated quality scores or even someone else's response. 
        Last, the RA is considered to be semi-honest, but may try to track a worker across tasks.
        Based on the above, Table~\ref{tab:attacks} presents possible attacks to systems operating in our threat model.

    \begin{table}[t]
    \footnotesize
        \centering
        \begin{tabular}{p{0.1\columnwidth}p{0.75\columnwidth}}
            \toprule
            \textbf{Entity} &  \hspace{7em} \textbf{Adversarial Behavior} \\
            \midrule
            \rotatebox[origin=c]{45}{{\multirow{3}{*}{\raisebox{2pt}{\hspace{-1.6em} Requester}}}}  & -- ``de-anonymize'' workers from multi-task participation \\[-0.7em] 
            &  -- update qualities and calculate payments arbitrarily\\
            &  -- avoid paying or updating qualities\\[2pt]
            \midrule
            \rotatebox[origin=c]{45}{{\multirow{2}{*}{\raisebox{2pt}{\hspace{-1em} Worker}}}} & -- participate in a single task multiple times\\[-0.5em]
            &  -- use other worker's answer/quality and get rewards \\[0pt]
            \midrule
            \rotatebox[origin=c]{45}{{\multirow{2}{*}{\raisebox{8pt}{\hspace{0.5em} RA}}}}    &        -- track a worker across tasks \\[-4pt]
            \bottomrule
        \end{tabular}

    \caption{Possible attacks based on \avecq's threat model}

    \label{tab:attacks}
    \end{table}
    
    \medskip
    \noindent\textbf{Problem Statement.}
        Considering the above system and threat model, the problem we study in this paper is the following: How to design an efficient and scalable system that carries out arbitrary crowdsourcing tasks and supports worker qualities, while safeguarding against the mentioned attacks.


    \begin{figure*}[t]

        \fbox{
            \begin{minipage}{2.03\columnwidth}
            	\centering
            	\begin{overpic}[width=1.02\columnwidth]{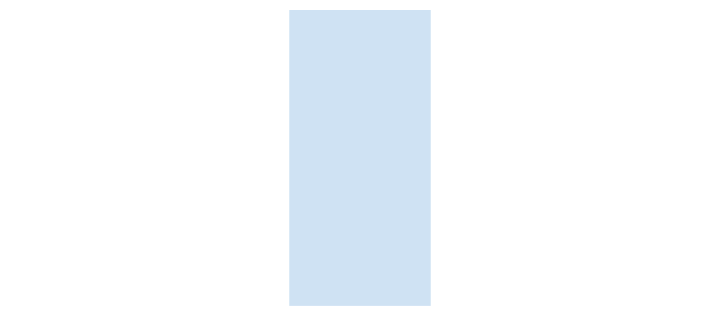}
            	\linethickness{0.6pt}
            	
                %
                %
                
            	\put(65,205){{\textbf{Requester} $R$} }
            	\put(225,205){{\textbf{Blockchain}}}
            	\put(380,205){{\textbf{Worker} $w_i$ } }

                \qbezier[200](5,202)(300,202)(485,202)
            
                \qbezier[100](5,202)(5,100)(5,7)
                \qbezier[100](485,202)(485,100)(485,7)
                \put(-5,150) {\footnotesize \rotatebox{-270}{\textbf{Task Creation}}}
            
                \put(25,190) {\footnotesize Deploys \texttt{CSTask}}
            	\put(125,190){\color{black}\vector(1,0){71}}
            	\put(22, 186){\color{black}\line(0,1){12}}
            	\put(7,188) {{\Circled[inner color=red, outer color=red]{1}}} 
            		\put(205,184)
                {\includegraphics[width=.03\textwidth]{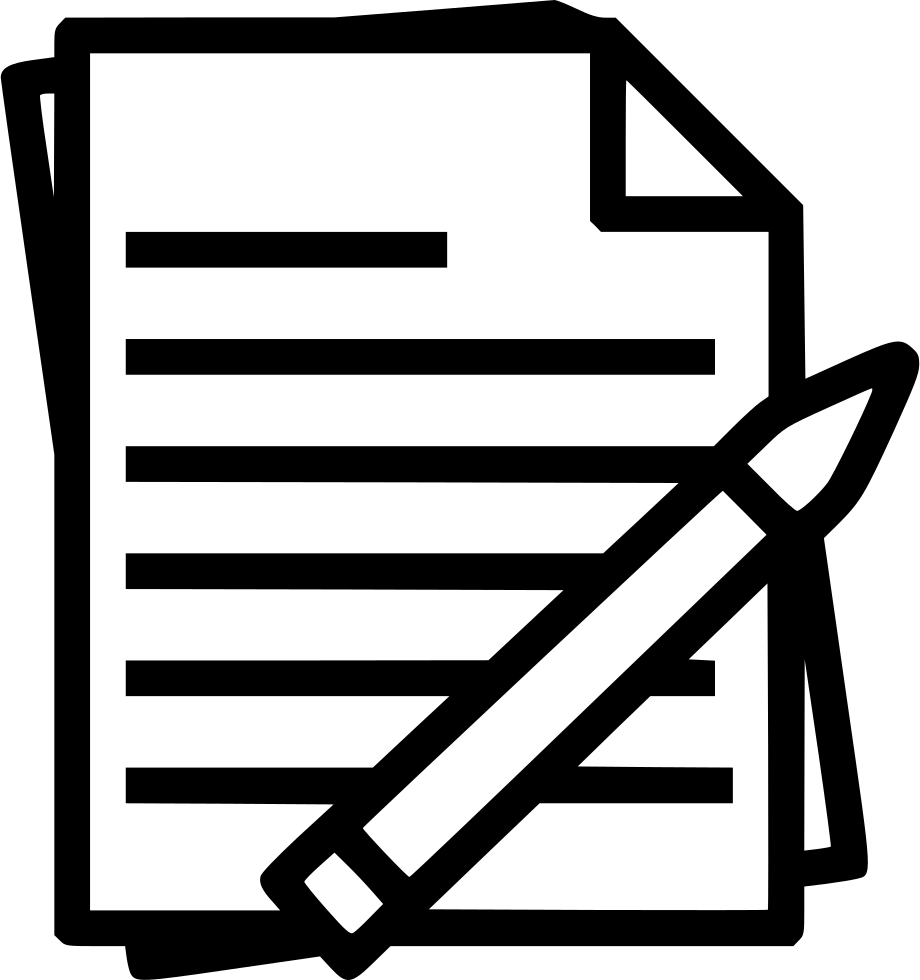}}

                \put(25,170) {\footnotesize Creates task with \textsc{CreateTask}\big($\tau$ =}
                \put(25,160) {\footnotesize $\langle d^\tau, \mathcal{A}^\tau, \Gamma^\tau, \mathcal{DL}^\tau, n_{th}^\tau, pk_R, P^\tau\rangle$\big)}
                \put(22,158){\color{black}\line(0,1){20}}
            	\put(170,172){\color{black}\vector(1,0){25}}
            	\put(7,165){{\Circled[inner color=red, outer color=red]{2}}}
            	\put(205,165)
                {\includegraphics[width=.03\textwidth]{write.png}}
            

                {\qbezier[100](5,152)(100,152)(197,152)}

                
                \put(340,180) {\footnotesize Pulls the task policy and $pk_R$ }
                \put(340, 170) {\footnotesize Produces $\pi_{o_i}^\tau$ with \textsc{ProveQual}}
                \put(335, 166){\color{black}\line(0,1){20}}
                \put(320, 175){{\Circled[inner color=blue, outer color=blue]{3}}}
                \put(300,168){\color{black}\vector(1,0){30}}
                	\put(278,160)
                {\includegraphics[width=.035\textwidth]{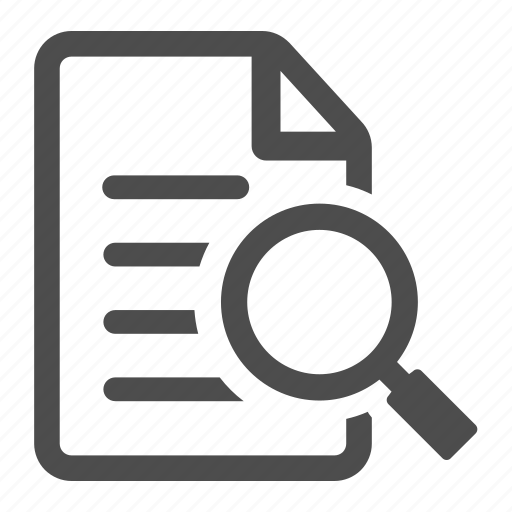}}
            
                \put(340, 150) {\footnotesize Submits a response with } 
                
                \put(340, 140) {\footnotesize \textsc{SubmitResponse}\Big($o_i^\tau=\langle E(pk_R,a_i^\tau; \cdot),$
                
                } 
                \put(340, 130) {\footnotesize  $ E(pk_R,pa_i^\tau ; \cdot), \textsf{Com}(q_i^{\tau-1},r_{c,i}^{\tau-1}+r_{\star,i}^\tau),$
                }
                \put(340, 120) {\footnotesize $E(pk_R,r_{k,i}^{\tau};\cdot),H\big(\textsf{Com}(q_i^{\tau-1},r_{c,i}^{\tau-1}),m_i\big),$ }
                    \put(340, 110) {\footnotesize $ \pi_{o_i}^\tau\rangle\Big)$ }
                \put(335, 108) {\color{black}\line(0,1){50}}
                \put(320, 140){{\Circled[inner color=red, outer color=red]{4}}}
                \put(330, 125){\color{black}\vector(-1,0){30}}  
                	\put(278,118)
                {\includegraphics[width=.03\textwidth]{write.png}}
                
                
                \qbezier[100](305,100)(395,100)(485,100)
                \put(486,191){\footnotesize \rotatebox{-90}{\textbf{Response Submission}}}
                
                \put(25,140) {\footnotesize Verifies $\pi_{o_i}^\tau~\forall w_i$ }
                \put(25,130) {\footnotesize Decrypts $E(pk_R,a_i^\tau; \cdot),\forall w_i$} 
                \put(25,120) {\footnotesize Calculates $a_\phi^\tau$ using \textsc{AnsCalc}}
                \put(195,125) {\color{black}\vector(-1,0){40}}
                \put(22,117) {\color{black}\line(0,1){30}}
                \put(7,110) {{\Circled[inner color=blue, outer color=blue]{5}}}
                    	\put(205,120)
                {\includegraphics[width=.035\textwidth]{read.png}}

                \put(25,110){\footnotesize Correctness proof $\pi_{a_\phi}^\tau$ with \textsc{AuthCalc}}
                \put(25,100) {\footnotesize Generate $q_i^\tau\mbox{~and~}\pi_{q_i}^\tau, \forall w_i$ using \textsc{AuthQual}}
                \put(25,90) {\footnotesize Comparison proof $\pi_{a_i}^\tau$ with \textsc{AuthValue}}
                \put(25,80) {\footnotesize Appends $\textsf{Com}(q_i^\tau,r_{c,i}^\tau)$ to $MT$
                }  
                \put(22,77) {\color{black}\line(0,1){40}}

                \put(25, 65) {\footnotesize Uploads $E(pk_r,a_\phi^\tau;\cdot)$ and $\pi_{a_\phi}^\tau$ } 
                \put(25, 55) {\footnotesize Submits qualities individually with} 
                \put(25, 45) {\footnotesize \textsc{SubmitQuality}$\big(H(r_{k,i}^{\tau}), r_{R,i}^\tau+r_{k,i}^\tau,r_{\star\star,i}^\tau+r_{k,i}^\tau,$}
                \put(25, 35) {\footnotesize $ \textsf{Com}(q_i^{\tau},r_{c-d,i}^{\tau}),pos_{q_i}^{\tau}+r_{k,i}^\tau,\pi_{q_i}^\tau, \pi_{a_i}^\tau\big),\forall w_i$}
                \put(25,25) {\footnotesize Submits payments individually with} 
                \put(25,15) {\footnotesize\textsc{WorkerPayment}$(Paym_i^\tau=\langle p_i^\tau,pa_i^\tau\rangle)$}
                \put(22,13) {\color{black}\line(0,1){60}}
                \put(7, 37) {{\Circled[inner color=red, outer color=red]{6}}}
            	\put(160,58){\color{black}\vector(1,0){35}}
            	  	\put(205,50)
                {\includegraphics[width=.03\textwidth]{write.png}}
            
            
                \qbezier[200](5,7)(300,7)(485,7)
                \put(-5,36){\rotatebox{-270}{\footnotesize\textbf{Response Processing}}}
                
                \put(230, 42) {\footnotesize \texttt{CSTask} pays}
                \put(230, 32) {\footnotesize$p_i^\tau$ to $w_i, \forall w_i$}
                \put(210, 35) {{\Circled[inner color=red, outer color=red]{7}}}
                \put(225, 29){\color{black}\line(0,1){22}}
            
                \put(330, 60) {\footnotesize Verifies proofs $\pi_{q_i}$ and $\pi_{a_i}^\tau$}
                \put(330,50) {\footnotesize If $q_i$ updated correctly:}
                \put(330, 40) {\footnotesize $~\implies$ Adopts $q_i$ as new quality}
                \put(330, 30) {\footnotesize Else:}
                \put(330, 20) {\footnotesize $~\implies$ Approaches RA with $q_i,\pi_{q_i}^\tau$ and $\pi_{a_i}^\tau$ }
            	\put(325, 15){\color{black}\line(0,1){52}}
            	\put(300, 22){\color{black}\vector(1,0){20}}
                \put(310, 33){{\Circled[inner color=blue, outer color=blue]{8}}}
                \put(470,35) {\rotatebox{-90}{\footnotesize\color{magenta}{Protest}}}
                \textcolor{magenta}{\qbezier[55](329,37)(398,37)(468,37)}
                \textcolor{magenta}{\qbezier[15](326,36)(326,23)(326,10)}
                \textcolor{magenta}{\qbezier[15](463,36)(463,23)(463,10)}
                \textcolor{magenta}{\qbezier[55](324,10)(385,10)(458,10)}
                \put(275,15) {\includegraphics[width=.035\textwidth]{read.png}}
            
                \put(479,90){\footnotesize\rotatebox{-90}{\textbf{Quality Verification}}}
            
                \put(340,-3) {\includegraphics[width=.02\textwidth]{read.png}}
                \put(353, -1) {\footnotesize Read}
                \put(290,-2) {\includegraphics[width=.018\textwidth]{write.png}}
                \put(303, -1) {\footnotesize Write}
                \put(390,1) {{\Circled[outer color=red]{}}}
                \put(397, -1) {\footnotesize On-chain}
                \put(440,1) {{\Circled[outer color=blue]{}}}
                \put(447, -1) {\footnotesize Off-chain}
    	
    	       \end{overpic}
    	   \end{minipage}
    	   }

    	\caption{$\mathsf{AVecQ}$: Task-specific stage analysis. \
        To initiate a task, a requester $R$ deploys a smart contract.
        A worker may submit a response including her (encrypted) answer, latest quality (commitment), and attest to the validity of the quality and conformity to the task policy (via a ZKP). 
        $R$ then verifies all submitted proofs off-chain, computes the final answer $a_\phi^\tau$, quality updates (commitments), payments, and ZKPs, and uploads all but $a_\phi^\tau$ on-chain.
        Last, workers get rewards and adopt their new qualities.
        }
    
    	\label{fig:avecq-Pro}
    \end{figure*}

%
%

\section{\avecq\label{sec:avecq}}

    This section presents our crowdsourcing system.
    Our construction utilizes cryptographic components (i.e. digital signature, public-key encryption, commitment scheme, zero-knowledge proof protocol, Merkle tree, and hash function) and a smart contract.
    Particularly, to create a task $\tau$, a requester can specify a set of attributes depending on the expressivity of the task. 
    These include:
    \emph{(i)} the task description $d^\tau$, 
    \emph{(ii)} the set of available answer choices $\mathcal{A}^\tau = \{ \textrm{a}_{1}^\tau,\ldots, \textrm{a}_{c}^\tau \}$, 
    \emph{(iii)} the maximum budget $\Gamma^\tau$ the requester is willing to allocate to the workers,
    \emph{(iv)} a set of deadlines $\mathcal{DL}^\tau$ (signifying until when workers may submit their responses and until when task processing must be completed by the requester), 
    \emph{(v)} a threshold number of workers $n_{th}^\tau$ (denoting the minimum participation of workers that is needed to calculate the final answer), 
    \emph{(vi)} the requester's public key $pk_R$ (which the workers will use to encrypt their sensitive data), 
    and \emph{(vii)} the task policy $P^\tau$ (including task-related information e.g., quality-threshold participation requirements, the final answer calculation mechanism, the quality-update rule, and the payment scheme).
    We denote the set of responses workers submit to task $\tau$ as $\mathcal{O}^\tau_{n^\tau} = \{o_{1}^{\tau}, \ldots, o_{n^\tau}^{\tau} \}$, where $n^\tau$ denotes the the total workers who provided responses to $\tau$.
    A response $o_i$ of worker $w_i$ includes, among others, its answer to the task ($a_i^\tau$) and its quality from the previous task ($q_i^{\tau-1}$)\footnote{We acknowledge that the ``previous'' task for two workers $w_i$  and $w_j$ might differ, however we use $q_i^{\tau-1}$ and $q_j^{\tau-1}$ to denote their latest qualities.}.  
    Therefore, we denote the set of answers to the task $\tau$ as $A_{n^\tau}^\tau=\{a_1^\tau,\ldots,a_{n^\tau}^\tau\}$  and the set of corresponding qualities as $Q_{n^{\tau}}^{\tau-1}=\{q_{1}^{\tau-1},\ldots,q_{n^\tau}^{\tau-1}\}$.

    During a task, the requester collects $\mathcal{O}^\tau_{n^\tau}$ and extracts $A_{n^\tau}^\tau$. 
    Then, (if needed) it calculates the final answer and for all participants their updated qualities and payments. For this, the requester uses the following algorithms:
    \begin{enumerate}[leftmargin=*, topsep=0.1pt, itemsep=0pt]
        \item \textsc{AnsCalc}$(A_{n^\tau}^\tau,Q_{n^\tau}^{\tau-1},P^\tau)\rightarrow a_\phi^\tau$: On input the set of answers from the participating workers $A_{n^\tau}^\tau$, the qualities $Q_{n^\tau}^{\tau-1}$ and the participating policy $P^\tau$, it outputs the final answer $a_\phi^\tau$. 
        
        \item \textsc{QualCalc}$(Q_{n^\tau}^{\tau-1},a_\phi^\tau,P^\tau)\rightarrow Q_{n^\tau}^\tau$: On input the set of the quality scores from the participants $Q_{n^\tau}^{\tau-1}$, the final answer $a_\phi^\tau$, and the participation policy $P^\tau$, outputs the set of the updated quality scores for every participating worker $Q_{n^\tau}^\tau=\{q_1^\tau,\ldots,q_{n^\tau}^\tau\}$.
        
        \item \textsc{PaymCalc}$( A_{n^\tau}^\tau,Q_{n^\tau}^{\tau-1},a_\phi^\tau,P^\tau)\rightarrow Paym^\tau$: On input the set of the answers $A_{n^\tau}^\tau$ and qualities $Q_{n^\tau}^{\tau-1}$, the final answer $a_\phi^\tau$, and the participation policy $P^\tau$, outputs the set of payments for every worker that participated $Paym^\tau=\{p_1^\tau,\ldots,p_n^\tau\}$.
    \end{enumerate}

    \subsection{Protocol}
        Next we describe our construction and explain the design rationale.
        Particularly, our protocol comprises of two stages: a preprocessing-setup and a task-specific one. 
        To assist the reader, we provide the following task as a representative example and follow its execution through the rest of this subsection. 
        
        \smallskip\noindent\textbf{Running example ($\tau_{\texttt{ex}}$).}  
            A requester deploys on $31/1/2023$ at $23:00$ the following image annotation task: 
         
            \begin{align*}
                \mathcal{AT}=\{\text{``Does this image contain a duck?''},
                (\text{Yes},\text{No}),\\
                200ETH,
                (31/1/2023/23:30,1/2/2023/23:59),
                1001,\\
                \texttt{0x7584e47e7a7e09a6be64dc3aeaf0b64364234d9c},\\
                (\text{Quality}>75\%,\text{Majority},\text{Beta distribution},\\
                \text{Correct:$0.0001$ETH,Not correct:$0.00005$ETH})\}
            \end{align*}

        \noindent\underline{Remark:} This is the most expressive task which can be executed through $\mathsf{AVeCQ}$, in terms of supported features. 
        For tasks requiring less functionality (e.g. reward each worker horizontally or do not require/need qualities) the related elements/steps can be omitted.   
        
        \smallskip
        \noindent\textbf{Preprocessing-Setup stage.}
            Ahead of time, the RA generates the public parameters $PP$ by running the zk-SNARKs' setup and encryption key-generation algorithms. 
            When a worker $w_i$ wishes to participate in the crowdsourcing system, it presents the RA with unique identification data $m_i$. The RA, in turn, generates a participation certificate $cert_i$ (i.e., an EdDSA signature on $m_i$, which in the U.S.A. could be the Social Security Number of $w_i$.).
            Additionally, it initializes the quality of $w_i$, generates the commitment $\textsf{Com}(q_i^{(0)},r_{c,i}^{(0)})$ where $r^{(0)}_{c,i} \sample\mathbb{Z}_p$, and appends this as a leaf to a Merkle Tree $MT$. 
            Last, the RA provides $w_i$ with ($cert_i$, $q_i^{(0)}$, $r_{c,i}^{(0)}$). 
            Notably, this is a \emph{dynamic} stage, as new workers can register arbitrarily and independently of other operations. 
            In $\tau_{\texttt{ex}}$, the RA sets up the SNARKs for calculating the most popular answer $a^\tau_\phi$, the Beta distribution, and for the $a^\tau_i\overset{?}{=}a^\tau_\phi$ check.
            We formally present all four zk-SNARKs' construction \avecq utilizes in Section~\ref{subsec:zk-snarks} and further elaborate them in Appendix~\ref{sec:app-zk-snarks}. 

        \smallskip
        \noindent\textbf{Task-specific stage.}
            To carry out a specific task, the requester and participating workers engage in a protocol having \emph{four} phases: \emph{Task Creation, Response Submission, Response Processing, and Quality Verification}. 
            Broadly, a requester creates a crowdsourcing task by deploying a smart contract on a blockchain to elicit responses from interested workers. 
            Upon collecting all  responses, the requester invokes \textsc{AnsCalc} to extract the final answer. Additionally, the requester invokes \textsc{QualCalc} to calculate the updated worker qualities and \textsc{PaymCalc} to calculate individual compensations. 
            All participants have known public addresses and can connect to the blockchain network. 
            Figure~\ref{fig:avecq-Pro} depicts all task-specific phases of $\mathsf{AVeCQ}$'s protocol.
            We present all the phases in detail, below.
        
    \subsubsection{Task Creation}
        This phase includes solely on-chain computations. First, a requester deploys a smart contract (\texttt{CSTask})  and deposits $\Gamma^\tau$ funds into it. 
        Next, to publish a new task $\tau$ the requester (suppose Alice) uses the \textsc{CreateTask} method of \texttt{CSTask} and uploads on-chain the following transaction:
        $tx_{\textsc{CreateTask}(\tau)}=\langle d^\tau, \mathcal{A}^\tau,\Gamma^\tau,\mathcal{DL}^\tau,n_{th}^\tau, pk_R, P^\tau \rangle = \mathcal{AT}$, in $\tau_\texttt{ex}$. 
        The requester then awaits for the next phase to conclude.

    \subsubsection{Response Submission}
        An interested $w_i$ can invoke the \textsc{SubmitResponse} method of \texttt{CSTask} to submit its response.
        Naively, $w_i$ just needs to provide the requester with its answer, its quality, a proof about holding a valid quality, and a public address for compensation. 
        However, \textsc{SubmitResponse}-related data are uploaded to a blockchain. 
        Thus, if we allow $w_i$ to send this data in the plain some of our security properties will be trivially violated\footnote{See Section~\ref{sec:secanal} and Appendix~\ref{sec:proofs} for a more elaborate analysis.}.  
        
        To ensure no leakage of sensitive data, $w_i$ provides a response $o_i^\tau$ containing,
        \emph{(i)} an encryption of its answer $a_i^\tau$,
        \emph{(ii)} an encryption of its public address $pa_i^\tau$, 
        \emph{(iii)} a re-randomized commitment regarding its latest quality $q_i^{\tau-1}$, 
        \emph{(iv)} an encryption of a random value $r_{k,i}^\tau$, 
        \emph{(v)} a hash of the commitment in $MT$ with $m_i$, 
        and \emph{(vi)} a corresponding proof of correctness for all these values. 
        This phase includes off-chain and on-chain computations as $w_i$ generates commitments, encryptions, hashes, and proofs locally; and later on invokes the \textsc{SubmitResponse} method to upload $o_i^\tau$ on \texttt{CSTask}.

        \smallskip
        \noindent\textbf{Off-chain.} 
            First, $w_i$ calculates the ciphertexts for its answer, address, and a random value $r_{k,i}^\tau\sample \mathbb{Z}_p$: $E(pk_R,a_i^\tau;\cdot)$, $E(pk_R,pa_i^\tau;\cdot)$, $E(pk_R,r_{k,i}^\tau; \cdot)$\footnote{Crucially, $r^\tau_{k,i}$ will be used to hide the leaf-position of the worker's newly updated quality commitment.}.  
            Additionally, $w_i$ computes the commitment $\textsf{Com}(q_i^{\tau-1},$ $r_{c,i}^{\tau-1}+r_{\star,i}^\tau)$\footnote{$\textsf{Com}(q_i^{\tau-1},r_{c,i}^{\tau-1})$ is a leaf in $MT$.} and the hash $H\left(\textsf{Com}(q_i^{\tau-1},r_{c,i}^{\tau-1}),m_i\right)$\footnote{We utilize this hash as a \emph{tag} that ensures that $w_i$ cannot re-submit a quality without being detected.} where $r_{\star,i}^\tau\sample\mathbb{Z}_p$. 
            Now, using the proof generating algorithm of our \textsc{ProveQual} zk-SNARK, a worker $w_i$ produces a proof $\pi_{o_i}^\tau$.
            Specifically, $\pi_{o_i}^\tau$ attests to 
            \emph{(i)} $w_i$ having registered, 
            \emph{(ii)} the existence of a leaf in $MT$ that hides $q_i^{\tau-1}$,  
            \emph{(iii)} $q_i^{\tau-1}$ conforms to $P^\tau$, 
            and \emph{(iv)} $(m_i,q_i^{\tau-1})$ are included in the calculation of  $H\big(\textsf{Com}(q_i^{\tau-1},r_{c,i}^{\tau-1}),m_i\big)$.

        \smallskip
        \noindent\textbf{On-chain.} 
            A worker $w_i$ can use the \textsc{SubmitResponse} method of \texttt{CSTask} to submit a response $o_i^\tau$ in the form of the following transaction:
            $tx_{SubmitResponse}(o_i^\tau) = \big<E(pk_R,a_i^\tau; \cdot), E(pk_R,pa_i^\tau; \cdot),$ $\textsf{Com}(q_i^{\tau-1},r_{c,i}^{\tau-1} \text{$+$} r_{\star,i}^\tau), E(pk_R,r_{k,i}^\tau;\cdot),
            H\left(\textsf{Com}(q_i^{\tau-1},r_{c,i}^{\tau-1}),m_i\right), \pi_{o_i^\tau} \big>.$

    \subsubsection{Response Processing}
        During this phase, the requester computes $a_\phi^\tau$, computes and communicates (via the \textsc{SubmitQuality} method) the updated worker qualities, and initiates payments (via \textsc{WorkerPayment}).
        Notably, workers need to be certain that the requester updated qualities, and corresponding payments, based on $P^\tau$.
        Thus, requesters provide individual zk-SNARK proofs for correctly updating qualities. 
        Overall, the requester performs the following off-chain and on-chain computations.  
    
        \smallskip
        \noindent\textbf{Off-chain.} 
            After the response submission deadline (specified in $\mathcal{DL}^\tau$) has passed and $n^\tau \geq n_{th}^\tau$ workers have submitted responses, the requester uses the verification algorithm of \textsc{ProveQual} to verify all proofs $\pi_{o_i^\tau}$, individually.  
            Then, it calculates $a_{\phi}^\tau \leftarrow$  \textsc{AnsCalc}($Q_{n^\tau}^{\tau-1},A_{n^\tau}^\tau$,$P^\tau$), all updated qualities using \textsc{QualCalc}$(Q_{n^\tau}^{\tau-1},$ $a_\phi^\tau,P^\tau)$, and payments using \textsc{PaymCalc}$(Q_{n^\tau}^{\tau-1},A_{n^\tau}^\tau,a_\phi^\tau,P^\tau)$.
            Recall that quality updates must be verifiable, even when only the requester knows $a_{\phi}^\tau$. 
            To achieve this we follow the next three steps.
            
            First, the requester generates a proof of the final answer calculation, using the workers' on-chain responses and the task policy. Specifically, it uses the zk-SNARK \textsc{AuthCalc} to do so. 
            \textsc{AuthCalc} decrypts all encrypted answers and using $P^\tau$ calculates $a_\phi^{\tau}$.
            Second, the requester provides an individual ``proof of correctness'' for each worker $w_i$ regarding the correlation between $a_{\phi}^\tau$, $a_i^\tau$, and $P^\tau$, using \textsc{AuthValue}. 
            This proof includes a single decryption and comparison.
            Third, the requester uses $o_i^\tau$, $a_{\phi}^\tau$, and $P^\tau$ to generate an individual ``proof of correct quality update'' for each worker via \textsc{AuthQual}.
            This includes a decryption and comparison as before, and additionally the verification of the new quality, based on $P^\tau$. 
    
        \smallskip
        \noindent\underline{Quality Updates}.
            Recall that workers submit their qualities in the form of commitments. 
            To update the quality of $w_i$, the requester computes: 
            $\textsf{Com}(q_i^{\tau-1},r_{c,i}^{\tau-1}+r_{\star,i}^\tau) \circ \textsf{Com}(\mu, r_{R,i}^\tau)$, $r_{R,i}^\tau\sample \mathbb{Z}_p$, where $\mu\in\mathbb{Z}$ is the difference between the old and new quality of $w_i$ ($q_i^\tau=q_i^{\tau-1}+\mu$).
            Finally, for each $w_i$, the requester re-randomizes the newly generated qualities commitment using ``dummy'' commitments i.e., $\forall i \in [n^\tau], \textsf{Com}(0,r_{\star\star,i}^{\tau})$ s.t. $r_{\star\star,i}^\tau\sample \mathbb{Z}_p$. 
            Formally, the requester appends the commitment 
            $
            \textsf{Com}(q_i^\tau,r_{c,i}^\tau)=\textsf{Com}(q_i^{\tau-1},r_{c,i}^{\tau-1}+r_{\star,i}^\tau)\circ\textsf{Com}(\mu, r_{R,i}^\tau)\circ \textsf{Com}(0,r_{\star\star,i}^{\tau})
            $
            to the MT and we denote $r_{c,i}^{\tau}=r_{c,i}^{\tau-1}+r_{\star,i}^\tau+r_{R,i}^\tau+r_{\star\star,i}^{\tau}$. 
            We also denote 
            $
            \textsf{Com}(q_i^\tau,r_{c-d,i}^\tau)=\textsf{Com}(q_i^{\tau-1},r_{c,i}^{\tau-1}+r_{\star,i}^\tau)\circ\textsf{Com}(\mu, r_{R,i}^\tau)
            $
            as the quality commitment without the ``dummy'' commitment with
            $r_{c-d,i}^{\tau}=r_{c,i}^{\tau-1}+r_{\star,i}^\tau+r_{R,i}^\tau$.
            In $\tau_\texttt{ex}$, suppose Bob is a worker who has answered ``Yes'', along with the majority of the rest of the workers, meaning $a_\phi^\tau\text{=``Yes''}$. 
            Alice computes 
            $\textsf{Com}(q_{i,\alpha}^\tau,r_{c,i}^\tau)=\textsf{Com}(q_{i,\alpha}^{\tau-1},\cdot)\circ\textsf{Com}(1, \cdot)\circ\textsf{Com}(0,\cdot)$ and \break $\textsf{Com}(q_{i,\beta}^\tau,r_{c,i}^\tau)=\textsf{Com}(q_{i,\beta}^{\tau-1},\cdot)\circ\textsf{Com}(1,\cdot)\circ \textsf{Com}(0,\cdot)$.
    
        \smallskip
        \noindent\textbf{On-chain.}
            The requester invokes the  \textsc{SubmitQuality} method, to communicate the newly updated qualities to each worker $w_i$.
            The transaction is of the form: $tx_{SubmitQuality} = \big\langle H(r_{k,i}^\tau), r_{R,i}^\tau + r_{k,i}^\tau, r_{\star\star,i}^\tau+r_{k,i}^\tau, \textsf{Com}(q_i^\tau,r_{c-d,i}^{\tau})\footnote{The quality commitment without the ``dummy'' re-randomization.}, pos_{q_i}^{\tau}+r_{k,i}^\tau,\pi_{q_i}^\tau, \pi_{a_i}^\tau\big\rangle$. 
            It includes: 
            \emph{(i)} a worker index $H(r_{k,i}^\tau)$, 
            \emph{(ii)} the new quality randomness $r_{R,i}^\tau + r_{k,i}^\tau$, 
            \emph{(iii)} the randomness of the ``dummy'' commitment $r_{\star\star,i}^\tau+r_{k,i}^\tau$, 
            \emph{(iv)} the commitment of the updated quality $\textsf{Com}(q_i^\tau,r_{c,i}^{\tau})$,
            \emph{(v)} the position $pos_{q_i}^{\tau}+r_{k,i}^\tau$ of the $MT$ leaf storing $w_i$'s new commitment-quality,
            and \emph{(vi)} the proofs $\pi_{q_i}^\tau$, $\pi_{a_i}^\tau$ for \textsc{Authqual} and \textsc{AuthValue}. 
            Last, the requester submits individual payments $p_i^\tau$ through the \textsc{WorkerPayment} method and \texttt{CSTask} reimburses the workers. 
        
    \subsubsection{Quality Verification}
        Quality verification includes only off-chain computations.
        In fact, $w_i$ verifies $\pi_{q_i}^{\tau}$ and $\pi_{a_i}^\tau$ and adopts the updated quality.
        We analyze the case when proofs do not pass verification in~Appendix~\ref{sec:miscelaneous}.
    
        \medskip
            \noindent\underline{Remark: Below-Threshold Worker Participation:}
            Recall that a requester, upon creating a task can specify a minimum participation of $n_{th}^\tau$ workers.
            If $n^\tau<n_{th}^\tau$ the task is considered void and the requester compensates workers for any expenses already made, re-randomizes the quality commitments, and produces corresponding $\pi_{q_i}$; allowing workers to take part in their next task seamlessly.

        \setlength\fboxrule{0.001pt}
        \newcommand{\ffbox}[1]{%
          {
           \setlength{\fboxsep}{-2\fboxrule}
           \fbox{\hspace{1.2pt}\strut#1\hspace{1.2pt}}
          }
        }

    \subsection{zk-SNARKs}\label{subsec:zk-snarks}

        \begin{figure*}[t]
        \centering
        \footnotesize
            \begin{tabular}{|l|l|}
                \hline 
                    \raisebox{-1.5em}{\rotatebox{45}{\footnotesize\textbf{\textsc{ProveQual}}}}
                    &
                    \begin{tabular}{l}
                        \textcolor{blue}{Statement} $\Vec{x}_{PQ}$: $PP,P^\tau,root_{MT},pk_{R},pk_{RA},\textsf{Com}(q_i^{\tau-1},r_{c,i}^{\tau-1}+r_{\star,i}^{\tau-1}),$
                	    $H\big(\textsf{Com}(q_i^{\tau-1},r_{c,i}^{\tau-1}),m_i\big),E(pk_R,a_i^\tau;\cdot),$
                	    $E(pk_R,pa_i^\tau;\cdot)$ \\[0.4em]
                        \textcolor{blue}{Witness} $\Vec{\omega}_{PQ}$: $cert_i, m_i, q_i^{\tau-1}, r_{c,i}^{\tau-1}, r_{\star,i}^\tau,r_{\star\star,i}^{\tau-1},\textsf{Com}(q_i^{\tau-1},r_{c-d,i}^{\tau-1}),
                        a_i^\tau,pa_i^\tau, path_{q_i^{\tau-1}}$ \\[0.4em]
                	    \textcolor{blue}{Language} $\mathcal{L}_{PQ} = \Big\{$~$\Vec{x}_{PQ}$~$|$~$\exists~\Vec{\omega}_{PQ}~  	\mbox{s.t.}$
                	    $\textcolor{red}{\textsf{EdDSAver}}(pk_{RA} ;cert_i,m_i)~\land~$
                        $\textcolor{red}{\textsf{ValidEnc}}\big(E(pk_R,a_i^\tau);a_i^\tau\big)~\land$
                        $\textcolor{red}{\textsf{TaskVer}}(P^\tau;q_i^{\tau-1})~\land~$ \\[0.4em] \hspace{2em}$\textcolor{red}{\textsf{QualVer}}\Big(H\big(\textsf{Com}(q_i^{\tau-1},r_{c,i}^{\tau-1}),m_i)\big),$
                        $\textsf{Com}(q_i^{\tau-1},r_{c-d,i}^{\tau-1});Com(q_i^{\tau-1},r_{c,i}^{\tau-1}),r_{\star\star,i}^{\tau-1},r_{\star,i}^\tau,m_i\Big)~\land~$ \\[0.4em]
                        \hspace{2em}$\textcolor{red}{\textsf{HashComVer}}\Big(H\big(\textsf{Com}(q_i^{\tau-1},r_{c,i}^{\tau-1}),
                        m_i\big);$
                        $\textsf{Com}(q_i^{\tau-1},r_{c-d,i}^{\tau-1}),\textsf{Com}(0,r_{\star\star,i}^{\tau-1}), m_i\Big)~\land$ \\[0.4em]
                        \hspace{2em}$\textcolor{red}{\textsf{ValidEnc}}\big(E(pk_R,pa_i^\tau);pa_i^\tau\big)~\land$ 
                        $\textcolor{red}{\textsf{MTPathVer}}\big(root_{MT} ;
                        \textsf{Com}(0, r_{\star\star,i}^{\tau-1}),$
                        $\textsf{Com}(q_i^{\tau-1},r_{c-d,i}^{\tau-1}),path_{q_i^{\tau-1}},$ $H(r_{k,i}^{\tau-1}, m_i)\big)=1\Big\}.$\\[0.4em]
                \end{tabular} \\
                
                \hline
                
                    \raisebox{-1em}{\rotatebox{45}{\footnotesize\textbf{\textsc{AuthCalc}}}}
                    &
                    \begin{tabular}{l}
                        \textcolor{blue}{Statement} $\Vec{x_{AC}}$: $PP,E(pk_R,a_\phi^\tau;\cdot),\{E(pk_R,a_i^\tau; \cdot)\}_{i \in [n^\tau]}, pk_R$ \\
                        \textcolor{blue}{Witness} $\Vec{\omega}_{AC}$: $sk_R$\\[0.4em]
                        \textcolor{blue}{Language} $\mathcal{L}_{AC}\text{ = }\{ \Vec{x}_{AC}~|~ \Vec{\omega}_{AC} \mbox{~s.t.~}$
                        $\textcolor{red}{\textsf{ValidKeyPair}}(pk_R;sk_R)\land$ 
                        $\textcolor{red}{\textsf{FinAnsVer}}\big( P^\tau,\{E(pk_R,a_i^\tau; \cdot)\}_{i \in [n^\tau]},$
                        $E(pk_R,a_\phi^\tau; \cdot); sk_R\big)=1\}.$ \\[0.4em]
              \end{tabular}\\            
                \hline

                    \raisebox{-1em}{\rotatebox{45}{\footnotesize\textbf{\textsc{AuthQual}}}}
                    & 
                    \begin{tabular}{l}
                       \textcolor{blue}{Statement} $\Vec{x}_{AV}$: $PP,E(pk_R,a_\phi^\tau;\cdot),$  $E(pk_R,a_i^\tau;\cdot), pk_R$\\
                        
                       \textcolor{blue}{Witness} $\Vec{\omega}_{AV}$: $sk_R$\\[0.4em]
                        \textcolor{blue}{Language} $\mathcal{L}_{AV} = \{ \Vec{x}_{AV}~|~ \Vec{\omega}_{AV} \mbox{~s.t.~}$
                        $\textcolor{red}{\textsf{EqCheck}}\big(E(pk_R,a_\phi^\tau),E(pk_R,a_i^\tau),$
                        $pk_R;sk_R\big)\text{$\land$}$
                        $\textcolor{red}{\textsf{ValidKeyPair}}(pk_R;sk_R)\text{=}1\}.$\\[0.4em]

                \end{tabular} \\
        
                \hline

                    \raisebox{-0.8em}{\rotatebox{45}{\footnotesize\textbf{\textsc{AuthValue}}}}
                    & 
                    \begin{tabular}{l}
                        \textcolor{blue}{Statement} $\Vec{x}_{AQ}$: $PP,E(pk_R,a_\phi^\tau;\cdot),$ $E(pk_R,a_i^\tau;\cdot),$\
                        $\textsf{Com}(q_i^\tau,r_{c-d,i}^{\tau}), \textsf{Com}(q_i^{\tau-1},r_{c,i}^{\tau-1}+r_{\star,i}^{\tau}),pk_R$ \\
                        
                        \textcolor{blue}{Witness} $\Vec{\omega}_{AV}$: $sk_R$\\[0.4em]
                        
            
                        \textcolor{blue}{Language} $\mathcal{L}_{AQ}\text{ = }\{ \Vec{x}_{AQ}~|~ \Vec{\omega}_{AQ} \mbox{~s.t.~}$
                        $\textcolor{red}{\textsf{NewQual}}\big(E(pk_R,a_\phi^\tau),E(pk_R,a_i^\tau;\cdot),\textsf{Com}(q_i^\tau,r_{c-d,i}^{\tau}),$
                        $\textsf{Com}(q_i^{\tau-1},r_{c,i}^{\tau-1}\text{+}r_{\star,i}^{\tau}); sk_R\big)~\land$ \\[0.4em]
                        \hspace{10em}$\textcolor{red}{\textsf{ValidKeyPair}}(pk_R;sk_R)=1\}.$ \\[0.4em]
                \end{tabular}\\
                \hline
            
            \end{tabular}  
     
            \caption{Statements, witnesses, and languages for \textsc{ProveQual,~AuthCalc,~AuthValue,} and \textsc{AuthQual}.}
            \label{fig:zk-snarks}
         
        \end{figure*}

        \avecq utilizes four different zk-SNARKs, allowing workers to participate in tasks honestly and requesters to calculate the final answer, updated qualities, and payments correctly, all of which verifiably.  
        Below we explain the functionality of all checks included in the employed zk-SNARKs and we include all statements, witnesses, and languages in Figure~\ref{fig:zk-snarks}.
        
        \smallskip
        \noindent\ffbox{\textsc{ProveQual.}} 
            To participate in $\tau$, worker $w_i$ uses \textsc{ProveQual}'s proving algorithm to generate a proof $\pi_{o_i}^\tau$. 
            \textsc{ProveQual} performs seven checks.
            \textsf{EdDSAVer} checks $w_i$ has already registered by verifying $cert_i$ signature using $m_i$, 
            while \textsf{MPathVer} that $q_i^{\tau-1}$ exists hidden as a commitment in a leaf of $MT$. 
            Additionally, \textsf{TaskVer} ensures $q_i^{\tau-1}$ conforms with $P^\tau$, 
            and \textsf{HashComVer} that $(m_i,q_i^{\tau-1})$ are included in the calculation of  $H\big(\textsf{Com}(q_i^{\tau-1},r_{c,i}^{\tau-1}),m_i\big)$. 
            Last, \textsf{ValidEnc} verifies well-formedness of the ciphertexts in $o_i^\tau$ and \textsf{Qualver} that $m_i$  and the same $q_i$ are included in the hash and the commitment.
            
        \smallskip
        \noindent\ffbox{\textsc{AuthCalc.}} 
            The requester uses $\textsc{AuthCalc}$ to attest to the correctness of $a_\phi^\tau$. 
            First, \textsf{ValidKeyPair} checks if the secret key in the witness consists a valid key-pair with the public key provided by the requester during the task's creation. 
            Afterwards, the circuit decrypts all workers' encrypted answers and computes the final answer based on $P^\tau$. 
            \textsf{FinAnsVer}($\cdot$)$=1$ if the computed final answer matches the decryption of the requester-submitted encrypted final answer. 
        
        \smallskip
        \noindent\ffbox{\textsc{AuthValue} and \textsc{AuthQual}.}  
            These zk-SNARKs attest to the correctness of each worker's answer and quality update, respectively. 
            \textsc{AuthValue} first checks the validity of the requester's key pair with \textsf{ValidKeyPair}. 
            Next, it uses \textsf{EqCheck} to check if the worker's encrypted response equals the final encrypted answer. 
            Likewise, \textsc{AuthQual} first uses \textsf{ValidKeyPair} to check the key pair's well-formedness. 
            Additionally, it uses \textsf{NewQual} to check if the updated qualities were computed correctly using the encrypted final answer and the worker's response. $P^\tau$ is hard-coded in the SNARK.


\section{\avecq: Core Security Properties\label{sec:secanal}}

    We now introduce definitions, theorems, and proofs for the three core properties of our system.  In Appendix~\ref{sec:proofs}, we provide an exhaustive analysis of all properties \avecq satisfies and how our system is safeguarded against various popular attacks.


        \begin{definition}[Anonymity with Unlinkabiltiy]\label{def:anon}
            A crowdsourcing system is anonymous with unlinkability if no PPT adversary $\adv$ has non-negligible advantage in the following game, $\mathcal{G}_{Anon}$.
            \begin{enumerate}[leftmargin=*,noitemsep]
                \item[$\bullet$] \textit{Initialization}: $\mathcal{A}$ specifies  parameters $n,\lambda$. The challenger $\mathcal{C}$ runs the certificate generation algorithm to register $n$ workers such that the maximum worker set for $\mathcal{G}_{Anon}$ is $\mathcal{W}_t=\{w_1,\cdots,w_n\}$ and samples a bit $b\sample \{0,1\}$.
                
                \smallskip
                \item[$\bullet$] \textit{Corruption Queries}: When $\adv$ issues such a query, it specifies a set of workers $\mathcal{W}_c\subseteq \mathcal{W}_t$, and $\mathcal{C}$ $\forall w_i \in \mathcal{W}_c$ provides all respective private information to $\adv$.
                
                \smallskip
                \item[$\bullet$] \textit{Task processing}: $\adv$ specifies a task $\tau$ and a corresponding worker set $\mathcal{W}_\tau$. $\mathcal{C}$, $\forall w_i\notin \mathcal{W}_c$, samples random answers and computes all necessary encryptions, commitments, and proofs using the information for all participating workers, and forwards the responses to $\adv$. $\adv$ computes and communicates to $\mathcal{C}$ all proofs regarding the participation of all $ w_i\notin \mathcal{W}_c$. If $\mathcal{C}$ cannot verify even one of these proofs the game halts.
                
                \smallskip
                \item[$\bullet$] \textit{Challenge}: $\mathcal{A}$ specifies a task $\tau$, two worker sets $\mathcal{W}_\tau,\mathcal{W}_\tau^\prime \subseteq \mathcal{W}_t-\mathcal{W}_c$, with $|\mathcal{W}_\tau|=|\mathcal{W}_\tau^\prime|$,
                and forwards $\mathcal{W}_\tau,\mathcal{W}_\tau^\prime$ to $\mathcal{C}$. If $b=0$ then $\mathcal{C}$ ``runs'' $\tau$ using $\mathcal{W}_\tau$ or uses $\mathcal{W}_\tau^\prime$ otherwise.
                Specifically, $\mathcal{C}$ computes all necessary encryptions, commitments, and proofs for the non-corrupted workers of the two sets and forwards corresponding responses to $\adv$. 
    
                \smallskip
                \item[$\bullet$] \textit{Finalization}: $\mathcal{A}$ sends  $b^\prime\in\{0,1\}$ to $\mathcal{C}$.
    
            \end{enumerate}
    
            \noindent$\mathcal{A}$ wins in $\mathcal{G}_{Anon}$ if $b^\prime=b$. A naive $\adv$, by sampling $b^\prime$ at random has a $\frac{1}{2}$ probability of winning. 
            A system is anonymous if~$\forall$ PPT $\mathcal{A}$, $Adv^{\mathcal{G}_{Anon}}(\mathcal{A})=\abs{\Pr[b^\prime=1|b=1]-\Pr[b^\prime=1|b=0]}\leq negl(\lambda)$.
        \end{definition}
    
        \begin{theorem}
            \emph{Assuming \textsf{Com} is computationally hiding, \textsc{ProveQual} computationally zero-knowledge, and $H$ collision and pre-image resistant, \avecq is anonymous as in Def.~\ref{def:anon}.}
        \end{theorem}
        
        \begin{proof}
            For an adversary to non-trivially identify whether ``two workers are the same'' or not, one of the following conditions must be true about the adversary: \emph{(i)} compromised the hiding property of the commitment scheme and accessed qualities in the plain, \emph{(ii)} found a collision in the hash function, or \emph{(iii)} compromised the zero-knowledge property of the underlying SNARK for \textsc{ProveQual} and accessed identity-revealing witnesses. 
            In $\mathcal{G}_{Anon}$ the challenge query requires the Challenger to execute a task $\tau$ with either of the following two worker sets $\mathcal{W}_\tau$ and $\mathcal{W}_\tau^\prime$, depending on the challenger bit. 
            Our proof follows a standard hybrid argument across 
            all possible selections of $\mathcal{W}_\tau$ and $\mathcal{W}_\tau^\prime$, specifically over their overlap regarding workers.
            
            We prove indistinguishability of the view of the adversary in $\mathcal{G}_{Anon}$, regardless of the challenger bit, through a series of Hybrid games over all possible $\mathcal{W}_\tau\cap\mathcal{W}_\tau^\prime$.
            We denote $\mathcal{H}_j$ the hybrid where $|\mathcal{W}_\tau\cap\mathcal{W}_\tau^\prime|=j$. 
            Note that, when  $\mathcal{W}_\tau\cap\mathcal{W}_\tau^\prime=n, Adv^{\mathcal{G}_{Anon}=\mathcal{H}_n}(\adv)=0$. 
            The only difference between the views of executing hybrids $\mathcal{H}_j$ and $\mathcal{H}_{j+1}$ lie in the computation of a commitment, a hash and a proof. 
            This means that the advantage a PPT adversary $\adv$ has in distinguishing between the two hybrids is $Adv^{\mathcal{H},j\overset{?}{\approx} j+1}(\adv)\leq Adv^\textsf{C-Hiding}(\adv)+Adv^{H\textsf{-CR}}(\adv)+Adv^\textsf{ProveQual-ZK}(\adv)$.
            Since by assumption the commitment scheme is computationally hiding, the hash function is collision-resistant, and ProveQual is computationally zero-knowledge, no PPT adversary can distinguish between these two views with non-negligible advantage.   
            To conclude the proof we apply this transformation $n$ times from $\mathcal{H}_0$ to $\mathcal{H}_n=\mathcal{G}_{Anon}$. Since $n$ is polynomially bound no PPT adversary has a non-negligible advantage in winning $\mathcal{G}_{Anon}$.
        \end{proof}

        \begin{definition}\label{def:frr}
          A crowdsourcing system is free-rider resistant if no PPT adversary $\adv$ has non-negligible advantage in the following game $\mathcal{G}_{FRR}$.
         \begin{enumerate}[leftmargin=*,noitemsep]
            \item[$\bullet$] \textit{Initialization}: $\mathcal{A}$ specifies  parameters $n,\lambda$. The challenger $\mathcal{C}$ runs the certificate generation algorithm to register $n$ workers such that the maximum worker set for $\mathcal{G}_{FRR}$ is $\mathcal{W}_t=\{w_1,\cdots,w_n\}$ and samples a bit $b\sample \{0,1\}$.
            
            \smallskip
            \item[$\bullet$] \textit{Corruption Queries}: When $\adv$ issues such a query, it specifies a set of workers $\mathcal{W}_c\subseteq \mathcal{W}_t$, and $\mathcal{C}$ $\forall w_i \in \mathcal{W}_c$ provides all respective private information to $\adv$.
            
            \smallskip
            \item[$\bullet$] \textit{Task processing}: 
            $\adv$ defines a task $\tau$, by specifying \textsc{AnsCalc}$^\tau$, \textsc{QualCalc}$^\tau$, and \textsc{PaymCalc}$^\tau$, with range $[minComp^\tau$,$maxComp^\tau]$ for the workers' compensations. $\adv$ forwards all this information to $\mathcal{C}$ who $\forall w_i\notin W_c$, samples random responses, computes all necessary encryptions and proofs using the information for all participating workers, and forwards all responses $\{o_i\}    \forall w_i \in \mathcal{W}_h =$ $\mathcal{W}_t-\mathcal{W}_c$ to $\adv$. 
            
            \smallskip
            \item[$\bullet$] \textit{Challenge}: $\adv$ specifies a task $\tau$ and $\mathcal{C}$ provides $\adv$ with the set of responses $\mathcal{O} = \{o_i\}$, for each $w_i \in \mathcal{W}_h$. 
            $\adv$ then forwards to $\mathcal{C}$ a response $o_j = \langle E^\prime(pk_R,a_{i_1}^\tau; \cdot), E^\prime(pk_R,pa_{i_2}^\tau; \cdot), \textsf{Com}^\prime(q_{i_3}^{\tau-1},r_{c,i_4}^{\tau-1}+r_{\star,i_5}^\tau), E^\prime(pk_R,r_{k,i_6}^{\tau};\cdot),H^\star\left(\textsf{Com}(q_{i_7}^{\tau-1},r_{c,i_8}^{\tau-1}),m_{i_9}\right),pi_{o_{i_{10}}}^{\tau^\prime}\rangle$, where $\exists i^\dag \in \{i_1,i_2\cdots,i_{10}\}$, $i^\dag\in \mathcal{W}_h$, and $\tau^\prime\neq\tau$.
            If $b=0$, $\mathcal{C}$ outputs $p_i^\tau=minComp^\tau$ to $w_i$ or runs \textsc{PaymCalc$(a_{i_1}^\tau,\cdot)$}$\rightarrow p_i^\tau$ otherwise. 

            \smallskip
            \item[$\bullet$] \textit{Finalization}: $\mathcal{A}$ sends  $b^\prime\in\{0,1\}$ to $\mathcal{C}$.

        \end{enumerate}

        \noindent$\mathcal{A}$ wins in $\mathcal{G}_{FRR}$ if $b^\prime=b$. A naive $\adv$, by sampling $b^\prime$ at random has a $\frac{1}{2}$ probability of winning. 
        A system is free-riding resistant if $\forall$ PPT $\mathcal{A}$, $Adv^{\mathcal{G}_{FRR}}(\mathcal{A}) =\abs{\Pr[b^\prime=1|b=1]-\Pr[b^\prime=1|b=0]}\leq negl(\lambda)$.
    \end{definition}
     
    \begin{theorem}
        \emph{Assuming that the PKE scheme $E$ is CPA-secure, the hash function $H$ is collision-resistant, the commitment scheme $Com$ is computationally hiding, and \textsc{ProveQual} is computationally sound, \avecq is free-riding resistant as in Def.~\ref{def:frr}.}
    \end{theorem}
    
    \begin{proof}
        To break Free-riding Resistance, an adversary has to extract information regarding $w_i$'s answer $a_i$ or quality $q_i$. 
        Violating the first one reduces to breaking the security of the PKE scheme $E$. 
        The second condition is more complex. Specifically, the adversary ``wins'' iff it can break the hiding property of $Com$, find a collision in $H$, or break the soundness of \textsc{ProveQual}. 
        To prove indistinguishability between the view of the adversary regardless of the challenger bit we use a hybrid analysis where we change one by one the witnesses into random values while simulating proofs.
        The total advantage of $\adv$ is $Adv^{\mathcal{G}_{FFR}}(\mathcal{A}) \leq Adv^{Com}(\mathcal{A})+Adv^{CPA-security}(\mathcal{A})+Adv^{H-collision}(\mathcal{A})+Adv^{SNARK-soundness}(\mathcal{A}) \leq negl(\lambda)$. 
    \end{proof}

    \begin{definition}[Policy Verifiability]\label{def:pv}
      A crowdsourcing system is policy-verifiable if both following conditions hold: 
       
      \noindent\emph{(i)} For any task $\tau_i$ with worker set $\mathcal{W}_i= 
      \{w_{i,1},\cdots,w_{i,n_i}\}$, no PPT adversary can output $a_\phi^{\tau_i,\prime} \nleftarrow \textsc{AnsCalc}(A_{n^{\tau_i}}^{\tau_i},P^{\tau_i})$, or $\{q_{i,1},\cdots,q_{i,n_i}\}\nleftarrow \textsc{QualCalc}(Q_{n^\tau},a_\phi^{\tau_i},P^{\tau_i})$, and accepting proofs $\pi_{a_\phi}^{\tau_i}$, $\{\pi_{a_{i,1}}^{\tau_i},\cdots, \pi_{a_{i,n_i}}^{\tau_i}\}$, $\{\pi_{q_{i,1}}^{\tau_i},\cdots, \pi_{q_i,n_i}^{\tau_i}\}$ with more than negligible probability.
    
      \noindent\emph{(ii)} For any task $\tau_j$ no PPT adversary can fabricate $q_i^\star$, for any $w_i$ with $q_i$,
      and an accepting proof $\pi_{o_i^\star}^{\tau_j}\leftarrow\textsc{ProveQual}(q_i^\star,\cdot)$, where $q_i^\star \neq q_i$ with more than negligible probability.
    \end{definition}
    
    \begin{theorem}
        Assuming \textsc{AuthCalc, AuthValue, AuthQual} and  \textsc{ProveQual} are computationally zero-knowledge, \textsf{Com} computationally hiding and $H$ collision-resistant, \avecq satisfies Policy Verifiability as in Def.~\ref{def:pv}.
    \end{theorem}
    
    \begin{proof}
        To break Policy Verifiability an adversary has to violate any of the two conditions. 
        In either case, the adversary essentially needs to provide convincing results and corresponding proofs that do not satisfy the relations in Figure~\ref{fig:zk-snarks} but still pass verification.
        Therefore, violating the first one reduces to breaking the soundness of \textsc{AuthCalc}, or \textsc{AuthValue}, or \textsc{AuthQual} or the binding property of \textsf{Com}. 
        The second condition is more straightforward. Specifically, the adversary ``wins'' iff it can break the soundness of \textsc{ProveQual} or can find a collision in $H$ with non-negligible probability.
        That would contradict the underlying assumptions and therefore no PPT adversary can break the policy-verifiability property of \avecq. 
    \end{proof}

%
%

\section{\avecq: Experimental Evaluation\label{sec:exp}}

    We implement a prototype of \avecq\footnote{The codebase is available at: \url{github.com/sankarshandamle/AVeCQ}.} and report its performance. 
    We deploy the sole on-chain component of $\mathsf{AVeCQ}$, \texttt{CSTask}, over Ethereum using Solidity~\cite{dannen2017introducing}.
    Additionally, we develop requester and worker Java applications that connect with the Rinkeby~\cite{rinkeby} and Goerli~\cite{goerli} test networks using the Web3j framework. 
    We use the Zokrates toolbox~\cite{eberhardt2018zokrates} for all zk-SNARKs implementation. Last, we perform all cryptographic operations over the ALT\_BN128 elliptic curve~\cite{reitwiessner2017eip}.
    As for the answer calculation policies, we consider \emph{(i)} Average (Avg), and \emph{(ii)} $\gamma$-Most Frequent ($\gamma$-MF) with $\gamma\in\{1,3\}$.
    
    First, we measure $\mathsf{AVeCQ}$'s on-chain costs in gas and USD. 
    Next, we examine the impact of varying gas prices on the verification time for th \textsc{SubmitResponse} method. 
    Then, we measure the computation time and communication size  of the off-chain components. 
    Finally, we include E2E analyses that report on the total time, space, and expenses required for the completion of three real-world tasks and compare, where possible, with state-of-the-art systems.
    
    \smallskip
    \noindent\textbf{Setup.} 
        We construct \texttt{CSTask} with bytecode size $3.8$ KB with Solidity and deploy on the Rinkeby and Goerli testnets with the Web3j framework. 
        We monitor the created transactions through Etherscan~\cite{etherscan}. 
        We create three different versions for \textsc{AuthCalc}, \textsc{AuthQual}, and \textsc{AuthValue} zk-SNARKs for our three crowdsourcing policies.
        For our off-chain cryptographic primitives, we use the Zokrates-accompanying pycrypto library.
        For hashing we use Pedersen hash~\cite{hopwood2016zcash} and ElGamal cryptosystem as the PKE scheme, which are zk-SNARK-friendly. We execute all off-chain-component experiments on a $40$-core server with Intel(R) Xeon(R) CPU E5-2640 v4 @2.40GHz and $1$ GB RAM per core.

    \subsection{On-chain Measurements} 
    
        \noindent\textbf{Gas Consumption and Costs.} 
            We measure the gas consumption for $\texttt{CSTask}$ deployment and method execution. 
            We map gas to USD, using the default gas price of $1$ Gwei = $1\times 10^{-9}$ ETH, base fee $5$ GWei\footnote{At the time of writing, the Ethereum Mainnet base fee was $\approx5$ GWei~\cite{ethgas}.} and the price of $1$ ETH = $1554.89$ USD (24/01/2023). Notably, deploying \texttt{CSTask}
            consumes $\approx 1.34$ million gas units and
            costs $\approx 12.49$ USD. Further, the method \textsc{CreateTask} consumes $363491$ gas units costing $3.39$ USD. \textsc{SubmitResponse} requires $394604$ units and costs the participating worker $3.68$ USD.  Uploading \textsc{AuthCalc} proof consumes $120772$ gas units and costs the requester $1.13$ USD.

        \smallskip
        \noindent\textbf{Priority Fee vs Transaction Processing Time.}
            Transaction processing time behaves in an ``inversely proportional'' manner to the chosen priority fee, $\delta$.
            We observe that depending on the urgency of the task, workers may opt to use gas prices other than the default. 
            Motivated by this, we examine the average verification time (in blocks) of \textsc{SubmitResponse} for each $\delta\in \{0.5,1,1.1,1.5,2,5,10\}$  across $200$ instances, and depict the result in Figure~\ref{fig:exp-subres}. 
            As shown, for a gas price of $1.1$ GWei, the average verification time on Rinkeby is $2.54\pm 0.7$ blocks ($\approx30$ secs). With Goerli, the verification time marginally increases to $3.52\pm0.8$ blocks ($\approx42$) seconds. 
            As expected, we observe that the verification time decreases as $\delta$ increases. 
            However, the increase is minimal ($\approx0.7$ block across both test networks) and the standard deviation remains almost constant, for prices ranging from $1.1$ GWei to $10$ GWei. 
            Importantly though, on Rinkeby, for gas prices $<1.1$ GWei, the difference in processing time is considerable and with high deviation. E.g., for $0.5$ GWei the average processing time is $8.68$ blocks, with a deviation of $\approx 2$ blocks. 
            
            The Ethereum Main network shows a similar trend. For instance, the verification time decreases from $\approx 10$ minutes with $\delta=0$ to $\approx 3$ minutes with $ \delta=1$ GWei~\cite{ethgas}. As each crowdsourcing task's time sensitivity is absorbed in its deadline, we believe that $1$ GWei is a reasonable priority fee. For shorter deadlines, a worker can accordingly increase its priority fee.

            \begin{figure}[t]
                \centering
         
                \includegraphics[width=\columnwidth,trim={55pt 35pt 80pt 260pt},clip]{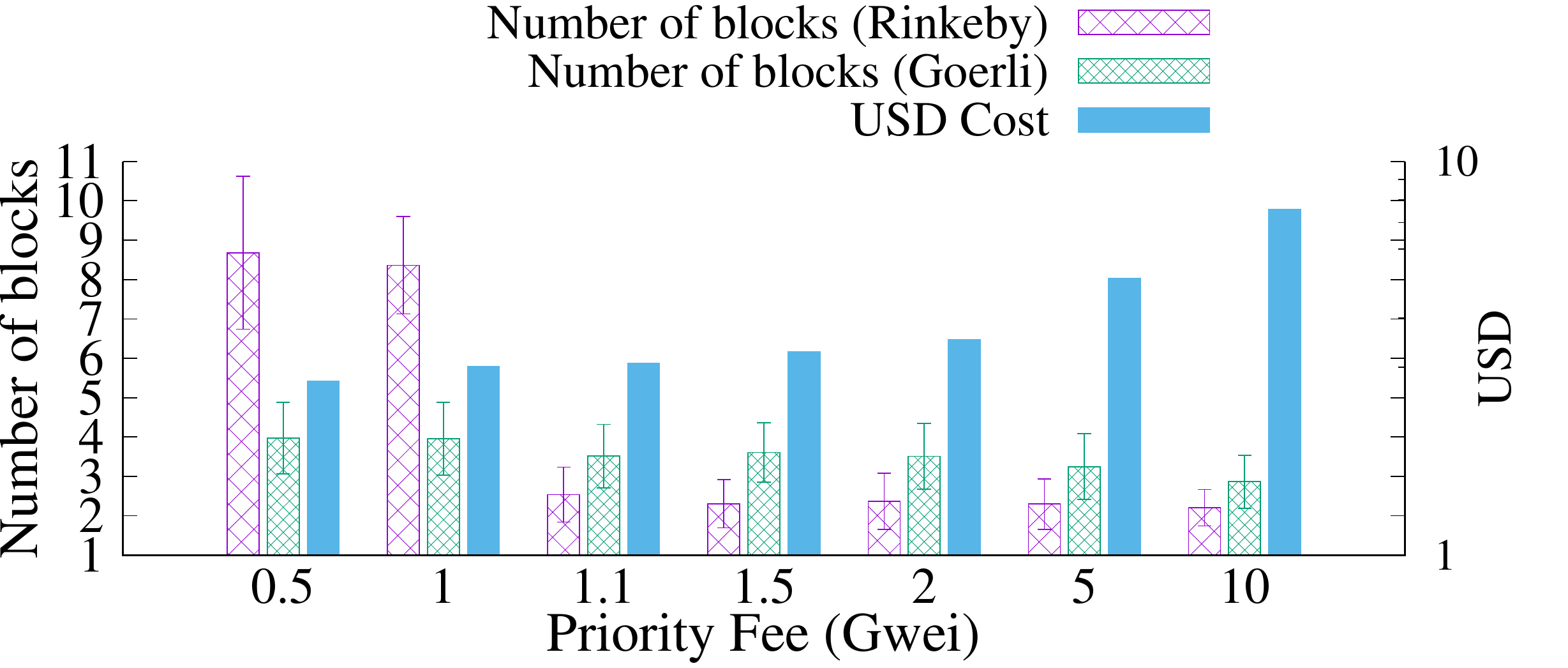}
           
                \caption{\textsc{SubmitResponse} Verification Time vs Priority Fee. Here, the base fee is 5 GWei. } 
             
                \label{fig:exp-subres}
                   \end{figure}
      \begin{figure}[t]
            \centering
            \includegraphics[width=\columnwidth]{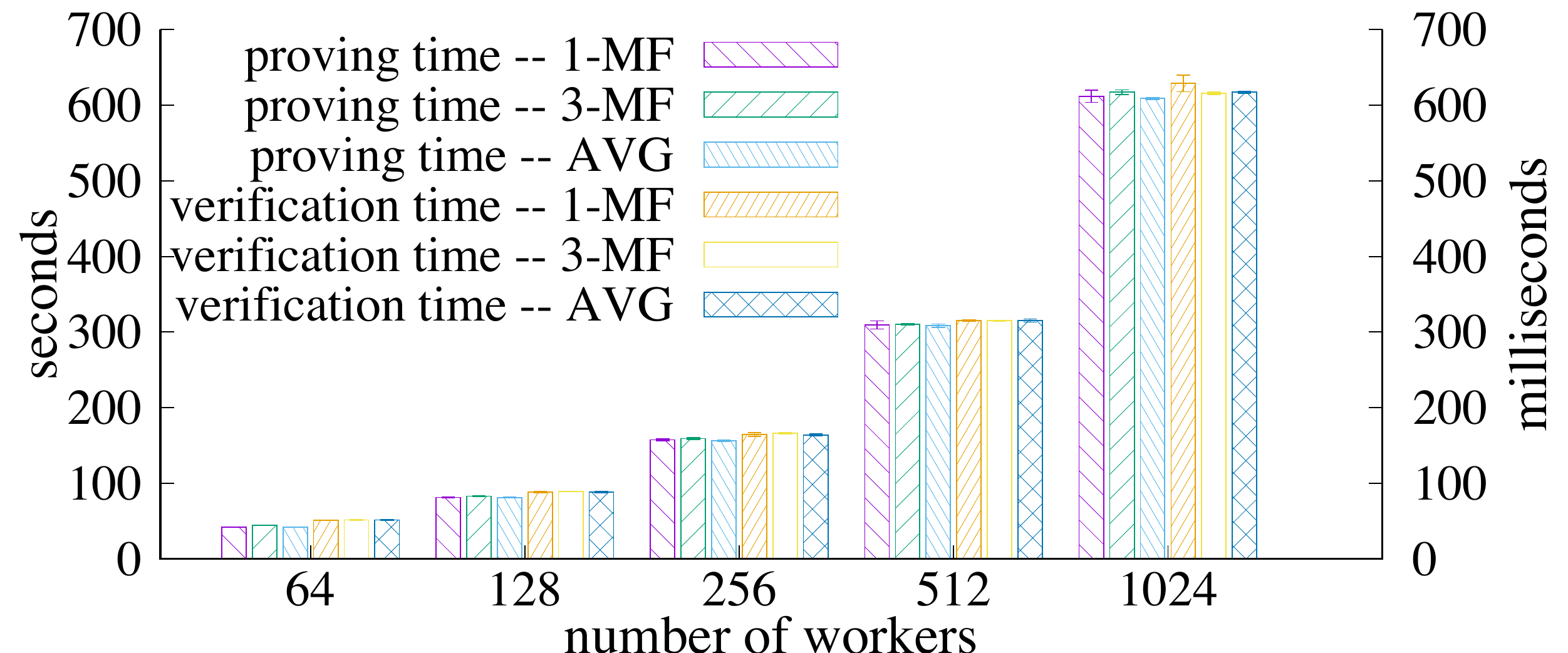}
         
            \caption{\textsc{AuthCalc}: Proving \& Verification Time vs. \break number of workers with $4$ choices.}
          
            \label{fig:exp-authcalc1}
           
            \end{figure}

    \subsection{Off-chain Measurements}

        \noindent\textbf{Non-SNARK Computations.}
            Both the requester and the workers generate ElGamal ciphertexts, Pedersen hashes and commitments.
            A single encryption takes $\approx11$ms, while a decryption takes $<1$ms. 
            Constructing a pre-image and computing a Pedersen hash requires $\approx 59$ms, while generating a Pedersen commitment takes $<1$ms. Overall, the communication size between the requester and a worker is $384$B and between a worker and the requester is $448$B.

        \smallskip
        \noindent\textbf{SNARKs Performance.} 
            Here we present measurements related to the proving and verification time for all employed zk-SNARKs. 
            As we show below for all three different applications and policies the proving time is in the order of seconds while the verification time in milliseconds.
            We report the average data across $10$ runs.
            
            \begin{figure}[t]
                     \includegraphics[width=0.9\columnwidth]{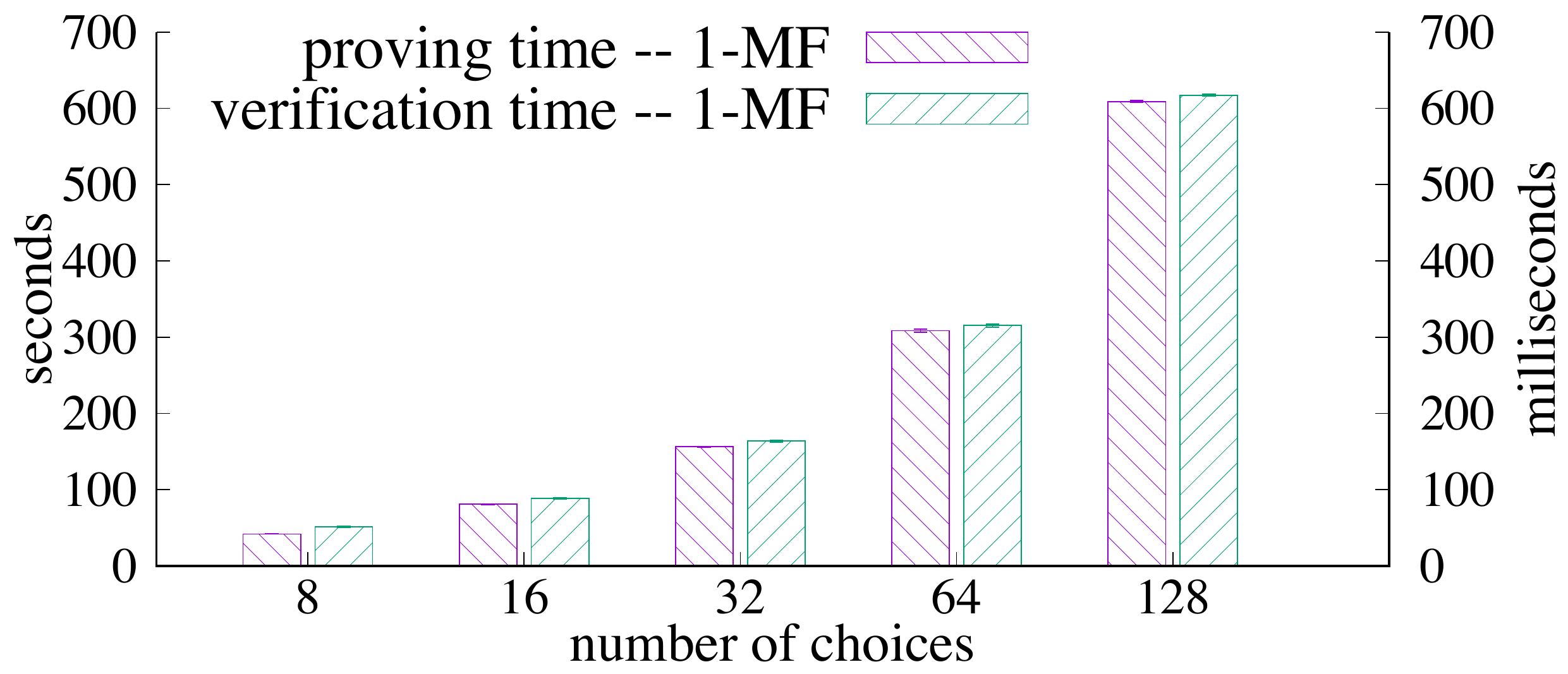}
                    \caption{\textsc{AuthCalc}: Proving \& Verification Time vs. number of choices for \textsc{AnsCalc}$=1$-MF with $1024$ workers.}
                    \label{fig:exp-authcalc2}
           \end{figure}
               \begin{figure}[t] 
                
                \includegraphics[width=0.9\columnwidth,trim={55pt 40pt 80pt 260pt},clip]{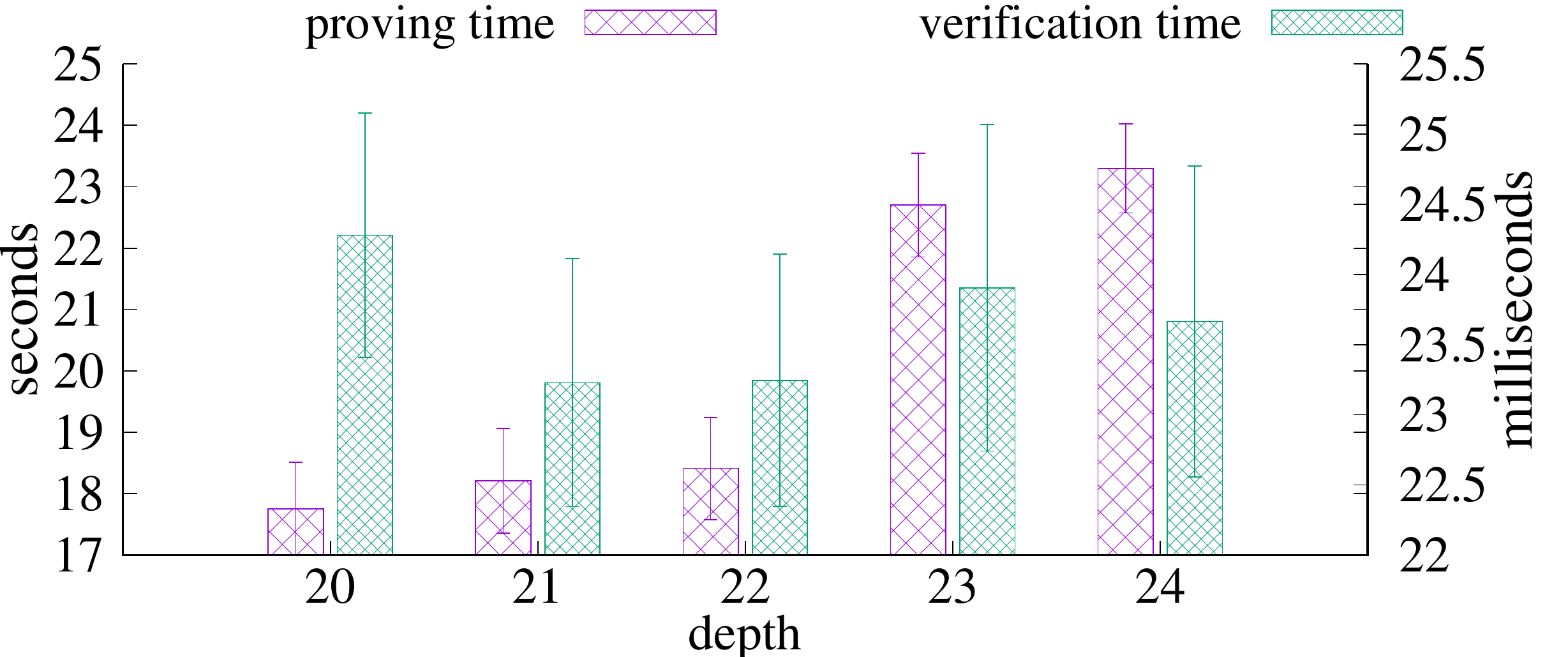}
             
                \caption{\textsc{ProveQual}: Proving \& Verification Time vs. $MT$ Depth.}
              
                \label{fig:exp-provequal}

            \end{figure}

        \smallskip
        \noindent\underline{\textsc{ProveQual}} 
            includes $8$ checks as described in Figure~\ref{fig:zk-snarks}, one of which is a membership proof for $MT$. 
            This is the only variable for the generation of $\pi_{o_i}^\tau$ and we examine how the Merkle tree depth affects the proving and verification times. 
            We present our findings for varying depths from $20$ to $24$ in Figure~\ref{fig:exp-provequal}. 
            Notably, to generate a proof for depth $23$ a worker needs $<23$ secs, which the requester verifies in $\approx24$ msecs.

        \smallskip
        \noindent\underline{\textsc{AuthCalc}} 
            embodies the implementation of \textsc{AnsCalc}. As such, we provide the implementation for the $\gamma$-Most Frequent (MF) and the Average (Avg) algorithms. 
            The performance of \textsc{AuthCalc} depends on the number of  \emph{(i)} workers and \emph{(ii)} choices. 
            Figure~\ref{fig:exp-authcalc1} depicts the performance of both implementations (for $\gamma=\{1,3\}$) with $4$ choices and workers varying from $64$ to $1024$. 
            Contrary, in Figure~\ref{fig:exp-authcalc2} we show the performance when fixing the number of workers to $1024$ and varying the choices from $4$ to $64$. 
            Crucially, the results confirm the practicality of our design, e.g., a requester using \avecq can output a proof pretending to be the most frequent answer, for $1$K workers and $64$ choices, in $<15$ mins. 
            We remark here that Avg is independent of the number of choices. Last, in all cases above, the verification time is $<0.7$ secs.

        \smallskip
        \noindent\underline{\textsc{AuthValue} and \textsc{AuthQual}}.
            The former produces a proof for the correctness of response $w_i$.
            For \textsc{AnsCalc}=$1$-MF we establish $a_i^\tau$ as correct if $a_i^\tau=a_\phi^\tau$, while for \textsc{AnsCalc}=Avg if it is within a $P^\tau$-specific range around the final answer.
            Now, \textsc{AuthQual} produces a proof regarding the corresponding update of $q_i$ based on the correctness of $a_i^\tau$.
            Recall that the requester does not have access to the qualities in the plain.
            To update a worker's quality, the requester ``adds'' a Pedersen commitment of $0$ or $1$ to the worker-provided commitment appropriately. 
            These zk-SNARKs have a constant number of constraints, and we provide the per-worker proof generation and verification times in Table~\ref{tab:2snarks}.

            \begin{table}[t]
                \centering
                \small
                \begin{tabular}{ccrr}
                    \toprule
                     \textbf{\textsc{AnsCalc}}    &  \textbf{zk-SNARK} & \textbf{Proving  (s)} & \textbf{Verification (ms)}  \\
                    \midrule
                    \multirow{2}{*}{\textbf{1-MF}}  &  \textsc{AuthQual} & $2.29$ & $11.8$	 \\ 
                    & \textsc{AuthValue}  &    $1.93$ & $11.2$ \\
                    \midrule
                    \multirow{2}{*}{\textbf{3-MF}}  &  \textsc{AuthQual} & $3.59$ & $12.3$	 \\ 
                    & \textsc{AuthValue}  &    $3.25$ & $11.7$ \\
                    \midrule
                    \multirow{2}{*}{\textbf{Avg}} &  \textsc{AuthQual}  & $2.91$ &	$12.1$ \\ 
                    & \textsc{AuthValue}  & $1.96$ &	$10.7$   \\
                    \bottomrule
                \end{tabular}

                \caption{\textsc{AuthQual} and \textsc{AuthValue} proving and verification time for different \textsc{AnsCalc} mechanisms.}
             
                \label{tab:2snarks}
            \end{table}

    \subsection{End-to-End (E2E) Run Time}
    
        To further demonstrate the practicality of $\mathsf{AVeCQ}$, we measure its performance on popular crowdsourcing tasks~\cite{krivosheev2020detecting,zheng2017truth,sun2018truth} simulated on real-world datasets~\cite{welinder2010multidimensional,he2016ups,perez2022dataset}. We measure
        each task's time from their deployment until completion. 
        
        \smallskip
        \noindent\textbf{Tasks.} 
            First, we consider Image Annotation~\cite{shah2015double},  commonly used  in crowdsourcing to generate datasets to train Machine Leaning (ML) models. 
            Similar to~\cite{krivosheev2020detecting}, we consider a task such as identifying whether a given image contains a duck or not. We deploy \texttt{CSTask}$_{ML}$ for this task. 
            Second, we consider the task of generating the Average Review of an online product through  \texttt{CSTask}$_{Review}$~\cite{krivosheev2020detecting}. 
            Last, we deploy \texttt{CSTask}$_{Gallup}$ to estimate public opinion through a Gallup poll. 
            E.g. to determine the ``COVID-19 fear level'' for citizens living in major Spanish cities~\cite{perez2022dataset}.

        \smallskip
        \noindent\textbf{Datasets.}  
            For \texttt{CSTask}$_{ML}$, we use worker reports obtained using a real-world image annotation dataset, Duck~\cite{welinder2010multidimensional}. Further, \texttt{CSTask}$_{ML}$ has \textsc{AnsCalc}=1-MF, $n_{th}^\tau=39$, and $|A^\tau|=2$. 
            For \texttt{CSTask}$_{Gallup}$, the worker reports are from the real-world survey conducted by P{\'e}rez \emph{et al.}~\cite{perez2022dataset}. Specifically, we have \textsc{AnsCalc}=$3$-MF, $|A^\tau|=5$ and we subsample $n_{th}^\tau=64$ worker reports from the $\approx$ 8K available. 
            Last, \texttt{CSTask}$_{Review}$ acquires an average review, with individual worker reports taken from the Amazon review dataset~\cite{he2016ups}. More concretely, \texttt{CSTask}$_{Review}$ has \textsc{AnsCalc}=Avg, where each worker can report a rating from the set $\{1,2,3,4,5\}$, i.e., $|A^\tau|=5$.
            We subsample $n_{th}^\tau=128$ worker reports from the $\approx 1M$ available.

        \begin{figure}[t]
            \centering
            \includegraphics[width=0.9\columnwidth,trim={55pt 40pt 80pt 230pt},clip]{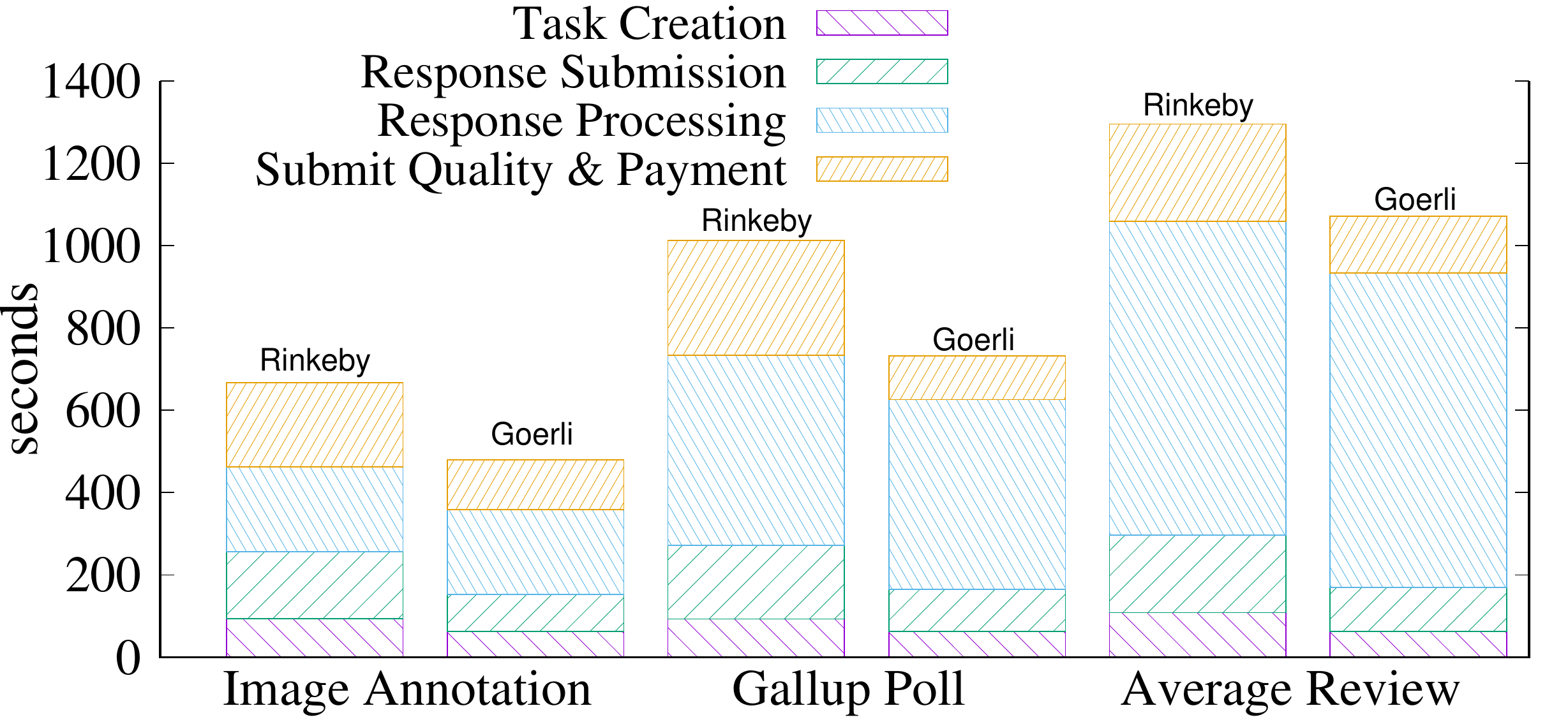}
               
            \caption{E2E completion time for different tasks with base fee 5 GWei and Priority fee 1 Gwei.}
            \label{fig:usecase}
            
        \end{figure}

        \smallskip
        \noindent\textbf{Results.} 
        Figure~\ref{fig:usecase} depicts our results with $\delta=1$ GWei.
        Depending on the task, a worker's response takes 93.34-108.52 secs (Rinkeby) and 62.36-63.12 secs (Goerli) to construct, submit, and get verified on-chain.
        Depending on the underlying crowdsourcing policy, the number of workers, and choices, to calculate the result and produce the corresponding proof the requester takes from $206$ to $764$ secs, while for uploading the quality updates and payments it takes 204-280 secs on Rinkeby and 106-140 secs on Goerli. Last, the quality verification phase takes $<1$~sec.

    \subsection{Comparing \avecq with Prior Works}

    We remark that \avecq outperforms state-of-the-art protocols~\cite{li2018crowdbc,lu2018zebralancer,duan2019aggregating,lu2020dragoon} in multiple metrics, despite incorporating worker quality and achieving stronger security properties. 
    Below we compare wherever and however possible with such systems. We  discuss initially E2E performance and then comparable individual aspects.

    \smallskip
    \noindent\textbf{E2E Comparison.}
    Surprisingly, very few prior works report the computational expenses or other overheads for an E2E execution of a specific task holistically. 
    In fact, only~\cite{lu2018zebralancer} and~\cite{duan2019aggregating} report such data. 
    Table~\ref{tab:e2ecomptable}, includes a head-to-head comparison between our system and these works. 
    As shown, \avecq outperforms~\cite{lu2018zebralancer} in  time and~\cite{duan2019aggregating} in  gas required while achieving similar time efficiency.

    \smallskip
    \noindent\textbf{Gas Consumed.} 
        Task creation requires $\approx4\times$ and response submission $\approx7.4\times$ more gas in \cite{lu2020dragoon}, compared to $\mathsf{AVeCQ}$.
        Additionally, the deployment cost of the smart contracts in~\cite{li2018crowdbc,duan2019aggregating} is comparable to the one of \texttt{CSTask}. Whereas, worker responses require $\approx 2$ times more gas in \cite{duan2019aggregating}. Further, the authors of~\cite{li2018crowdbc} perform the majority of their protocol operations on-chain, which would incur prohibitive gas (and monetary) costs for tasks with a high number of workers, based on a similar approach followed in~\cite{agorapreprint}.

    \smallskip
    \noindent\textbf{Supporting Workers and Task Choices.} 
        Last, unlike \cite{lu2018zebralancer,lu2020dragoon} we provide results for significantly greater number of workers and task choices. 
        The authors in \cite{lu2018zebralancer} provide results up to $11$ workers and $2$ choices, and only $4$ workers in \cite{lu2020dragoon}. 
        To the best of our knowledge, \avecq is the first anonymous crowdsourcing system measured against $>1$K workers and $>100$ choices (Figures~\ref{fig:exp-authcalc1} and \ref{fig:exp-authcalc2}).

        \begin{table}[t]
            \footnotesize
            \centering
            \begin{tabular}{ccc}
                \toprule
                \xrowht[()]{10pt}
                {\makecell{\textbf{Tasks}}} &   {\makecell{\textbf{Image Annotation}\\ ($|W|=39, |A|=2$)}}  & {\makecell{\textbf{Average Review}\\ ($|W|=128, |A|=5$)}}  \\
                \midrule
                \xrowht[()]{10pt}
                \makecell{\textbf{\avecq}}  & \makecell{Time: $<8$mins \\ Gas: $<19$MWei}   &  \makecell{Time: $<18$mins \\ {Gas:} $<55$MWei}\\
                \hdashline 
                \xrowht[()]{25pt}
                \makecell{\textbf{E2E Comparable}\\ \textbf{Blockchain-based} \\ \textbf{Systems}} &   \makecell{ZebraLancer~\cite{lu2018zebralancer} \\ Time: $>7$h$^\star$ \\ Gas: No data} & \makecell{Duan et al.~\cite{duan2019aggregating} \\ Time: $30$min$^\star$ \\ Gas: $>416$MWei$^\star$} \\[-5pt]
                \bottomrule
                \multicolumn{3}{c}{$^\star$ denotes extrapolation of results. }
            \end{tabular}

            \caption{Completion time of \avecq (on Goerli) vs contemporary works in popular crowdsourcing tasks. \label{tab:e2ecomptable}}

        \end{table}

%
%

\section{Discussion \& Conclusion}\label{sec:discon}

    \noindent\textbf{Discussion.}  
        Using gold-standard tasks to evaluate the ``trustworthiness'' of a worker is quite popular~\cite{GaoWL19,goel2019deep,pace21,bhitVLDB22}. \avecq can support such tasks. In fact, for these types of tasks, all requester off-chain computations are almost non-existent since all workers know the final answer upon completion of the task and, therefore, can verify the correctness of their quality updates and payments. Therefore, \avecq operates in a richer setting functionality-wise, and specifically for this subset of cases, it is even more efficient.       
        On a different note, there indeed exist quality/reputation systems based on more complex algebraic relations than additions (e.g., products~\cite{Farm} or Gompertz function~\cite{KanaparthyDG22}).
        We acknowledge this and identify as an interesting future direction the construction of an even more general protocol that can incorporate sophisticated quality scores inexpressible via homomorphic commitments (e.g.,~\cite{Gom1,Farm,KanaparthyDG22}).
        Last, \avecq is task/policy/blockchain agnostic at its core. The only requirement is that the policy can be expressed as a circuit and thus \textsc{AnsCalc}, \textsc{QualCalc}, and \textsc{PaymCalc} as zk-SNARKs.

    \smallskip\noindent\textbf{Conclusion.}
        In this work, we proposed $\mathsf{AVeCQ}$, the first anonymous crowdsourcing system with verifiable worker qualities. 
        $\mathsf{AVeCQ}$ leverages a fusion of cryptographic techniques and is built atop a blockchain that supports smart contracts. 
        Moreover, we demonstrated via extensive experimentation that our system is deployable in real-world settings. 
        Additionally, increases in the number of workers/choices for popular task policies do not impact the performance of $\mathsf{AVeCQ}$.
        In conclusion, \avecq outperforms state-of-the-art and guarantees stronger 
        security and privacy properties.

%
%


\bibliographystyle{plain}
\bibliography{ref}

%
%

\appendix

\section{Security Properties of \avecq}\label{sec:proofs}



    
    \subsection{Anonymity (with Unlinkability)}
    The goal is to define the property so that no entity can establish a connection between a worker's identity and its responses across tasks. 
    We therefore need to take a closer look to the response $o_i$.
    Remember that a worker $w_i$ submits $o_i^\tau$ to a task $\tau$ in the following form:
      $o_i^\tau=\langle E(pk_R,a_i^\tau;\cdot),E(pk_R,pa_i^\tau;\cdot),$ $\textsf{Com}(q_i^{\tau-1},r_{c,i}^{\tau-1}\text{+}r_{\star,i}^\tau),$ \break $E(pk_R,r_{k,i}^{\tau};\cdot),$ $H\big(\textsf{Com}(q_i^{\tau-1},r_{c,i}^{\tau-1}),m_i\big),\pi_{o_i}^\tau\rangle$.

    First, we look at the public address $pa_i$: 
    if is not unique across tasks, then any requester (or participant in general) can connect the same $pa_i$ with two discrete responses and thus identify-link a worker across two tasks. 
    Notably, this is the case regardless of whether the uploading mechanism is anonymous or not. 
    For example, in Ethereum a worker may select to respond to different tasks from equally different addresses.
    However, this measure alone is not sufficient. 
    In our system's case, by including $pa_i$ encrypted in the response tuple, we ensure that only the requester can identify any connection between the uploading mechanism and the public addresses that receive compensation after the task has been concluded.
    This, in turn, translates into the fact that \emph{any} other entity that is even monitoring \emph{all} blockchain transactions may connect each worker across \emph{at most} two tasks. 
    To avoid even such a case though, multiple approaches have already been proposed in the literature, from anonymous tokens e.g., ZCash~\cite{hopwood2016zcash} to mixnets or tumblers~\cite{tumblebit}; techniques easy to plug in $\mathsf{AVeCQ}$.
    
    We now consider which other element of the response can be used non-trivially to link a worker across two tasks.
    We observe that the quality scores can ``reveal'' the identity of a worker across two tasks.
    If a worker submits its quality in the response as a plaintext then any participant can possibly track them across tasks e.g., workers with high or low qualities.
    However, even if they are provided in a hiding manner i.e., as commitments in $\mathsf{AVeCQ}$, recall that it is actually the requester that updates the quality scores of the participating workers and then commits to them.
    Therefore, if qualities are used exactly as provided by the last requester then back-to-back responses can be linked across these two tasks, forming in fact the full chain or responses across \emph{all} tasks.
    To avoid this we require from workers to re-randomize the quality commitments before participating in the next task. 
    
    To prove that \avecq satisfies Anonymity we rely on the hiding property of \textsf{Com}, the collision-resistance of $H$, and the zero-knowledge property of \textsc{ProveQual}, assuming that no non-trivial information about $q_i$ is leaked to the requester or anyone else during response submission or processing. We refer to trivial leakage as information that one can decipher from the publicly available task policy (e.g., $q_i$ surpasses a certain $q_{th}^\tau$).
    We design a game $\mathcal{G}_{Anon}$ to formally capture anonymity as a property.
    At a high level, we state that a crowdsourcing system is anonymous if for any two distinct tasks $\tau_i,\tau_j$ with corresponding worker sets $\mathcal{W}_i, \mathcal{W}_j$, no entity can non-trivially distinguish whether $\exists~w_\star$ s.t. $w_\star \in \mathcal{W}_i$ and $ w_\star \in \mathcal{W}_j$.

    \smallskip
    \subsection{Free-rider Resistance}
    For this property we try and capture the following adversarial behavior: A worker whose goal is to utilize another worker's answer and/or quality and get compensated. We present and analyze all three cases below.

    First, workers might try to elude task execution or inflate total worker rewards\footnote{E.g., a worker who cannot clear the quality threshold for a task can collude with another worker, submit a duplicate response, and split the reward.} by submitting other workers' responses as their own. 
    Generally in such a case the requester would face the challenge of deciding which of the duplicate tuples was submitted first/legitimately. 
    However, since \avecq utilizes a blockchain all submitted responses are timestamped, meaning that even if a worker copies all parts of a response and only changes the public address part (or even submit the same response entirely), the requester can disregard any ``duplicate'' records.
    Thus, such behavior is counterincentivized as workers who pawn someone else's response as their own \emph{(i)} have to expend resources for uploading it on-chain and \emph{(ii)} \emph{deterministically} will get no reward.
    
    Second, workers might try to use someone else's quality to pass a specified threshold. 
    In $\mathsf{AVeCQ}$, when responding to a task, $w_i$ ties specifically its $q_i$ to the underlying $m_i$ of $cert_i$. 
    Then, each worker produces a proof of validity for all hashes, commitments, and encryptions in the submit-response tuple using \textsc{ProveQual}. From the soundness property of the employed zk-SNARK and the hiding property of the commitment scheme, no polynomially bound worker can produce a convincing proof  without a witnesses and from the zero-knowledge property of \textsc{ProveQual} no polynomially-bound worker can extract any underlying witness from $\pi_{o_i}^\tau$.

    Third, to ensure no worker can utilize the encrypted-answer part of another worker's response and use it with its own quality commitments we rely on the soundness and zero-knowledge property of \textsc{ProveQual}. Similarly to the above, no polynomially-bound worker can extract another worker's answer $a_i$ from its response $o_i$. 
    We design a game $\mathcal{G}_{FRR}$ to formally capture this property.
    At a high level, we state that a crowdsourcing system is secure against free-riding attacks if the adversary can submit a response $o_i$ containing \emph{any} element from another worker's response $o_i^\prime$ and can distinguish between getting the legitimate compensation based on the answer included in $o_i$ or the minimum possible compensation.

    \smallskip
    \subsection{Policy-Verifiable Correctness}
    No requester can produce convincing fabricated  $a_\phi^{\tau,\prime}$ and qualities.
    Additionally, no worker-provided response that contains an answer $a_i \not\in A^\tau$ is included in the $a_\phi^\tau$ calculation.
    To prove this, we rely on the soundness of the zk-SNARKs' outputs: $\pi_\phi^\tau \leftarrow$\textsc{AuthCalc}($\cdot$), $\pi_{q_i} \leftarrow \textsc{AuthQual}(\cdot)$, and $\pi_{a_i}^\tau \leftarrow \textsc{AuthValue}(\cdot)$. The trusted RA setups all SNARKS for policy $P^\tau$. Thus, no requester can produce a valid proof that pertains to a fabricated result 
    $a_\phi^{\tau,\prime}\neq a_\phi^\tau$ or qualities for a different policy $P^{\tau,\prime}$ without breaking the underlying zk-SNARK soundness property.

    \subsection{Other Properties\label{subsec:sec_B}}
    
    \smallskip
    \noindent\textbf{Sybil-Attack Resistance.}
    These attacks correspond to entities forging identities to participate in tasks. 
    For instance, a worker with a ``low'' quality may prefer to generate fresh identities to have higher chances of clearing quality thresholds and being able to participate in future tasks.
    Moreover, workers might attempt to participate in a task multiple times, under different identities, and reap rewards almost arbitrarily.
    On top of that, especially in tasks where the final answer is not known a-priori, a worker that participates arbitrarily-many times in a task can launch even more sophisticated attacks (see Section~\ref{sec:miscelaneous}). 
    The severity of Sybil attacks in crowdsourcing systems has been studied extensively and we rely on a trusted RA to combat them, similarly to prior works~\cite{lu2018zebralancer,lu2020dragoon,bhitVLDB22}.
    Recall that \avecq includes a registration phase where every worker provides a secret message $m_i$ and acquires a participation certificate in the form of an EdDSA signature $cert_i$ through the RA.
    To submit $o_i$, each worker has to generate a \textsc{ProvQual} proof, which includes a verification of $cert_i$ based on $m_i$. 
    From modeling the RA as trusted and the soundness property of \textsc{ProveQual}, no polynomially-bound worker can produce a convincing proof without the respective witnesses.

    \smallskip\noindent\textbf{Payment and Quality Deprivation Resistance.}
    No requester can avoid paying. 
    \texttt{CSTask} method \textsc{PayWorker} handles the payments of workers during the quality verification phase. Therefore, rightful reimbursements are allocated to workers, assuming an honest majority in the on-chain validators. 
    Any worker $w_i$ who submits a correct answer $a_i^\tau$ and yet is not paid for participating in $\tau$, during the protest period, can contact the RA with the evidence $\pi_{a_i}^\tau,o_i^\tau$. Upon verifying the proof, the RA can confiscate the requester funds, pay the worker, and impose added penalties (similarly, for when a requester does not upload on-chain a quality update).

    \subsection{Miscellaneous Attacks}\label{sec:miscelaneous}
    
    Below we analyze attacks that our crowdsourcing and threat model (Section~\ref{sec:model}) allow and possible mitigation techniques.
    
    \smallskip
    \noindent\textbf{Final-Answer Skewing.}
    First, since we allow arbitrary collusions between the entities and our system imposes no restrictions to the task policy, the following attack is enabled. Consider our running-example task where $|\mathcal{W}^\text{Yes}|=|\mathcal{W}^\text{No}|+1$. An adversarial requester can skew the final answer if he colludes with even just two additional workers. Interestingly, this type of attack in combination with certain task policies can even result in the requester giving out less rewards totally! 
    
    The following example is illuminating: Consider the case where the final answer is calculated as the average of the workers numerical responses and the qualities/rewards are calculated based on a proximity deviation between the worker's answer and the final answer. 
    In this case, the requester can collude with even just one worker who just needs to purposefully submit an ``overshot'' answer, affecting the final answer enough to reduce total expended rewards. 
    Our system does not safeguard workers against such game-theoretic attacks, which is in line with the broad mechanism design~\cite{KanaparthyDG22,PTS,BTS,PTSC,Farm,RBTS} and anonymous crowdsourcing~\cite{lu2018zebralancer,duan2019aggregating} literature. 
    Additionally, in tasks where the final answer of a task is calculated as a function of participating worker responses, workers may opt to misreport their answer to try and ``guess'' the final answer and get rewarded.
    Popular approaches to avoid such attacks is to impose constant rewards for all participating workers or enable tasks with only publicly available ground truth ~\cite{savage1971elicitation,lambert2009eliciting}. 
    In \avecq we adopt a more expressing model concerning the task policy, which includes the above approaches, but is not restricted to these.
    
    \smallskip
    \noindent\textbf{Quality Inflation.}
    Another possible (and rather subtle) threat stems from the fact that in \avecq workers verify just their own quality updates.
    In fact, if a worker realizes that its quality is \emph{lower} than what it should be or that the requester-provided proofs do not pass verification, it is incentivized to protest and correct the wrongdoing.
    However, if the updated quality is \emph{higher}, then the worker is incentivized to avoid protesting!  
    This in turn, enables requesters to collude with workers to raise arbitrarily their qualities.
    Previous blockchain-based works adopt on-chain verification to avoid similar issues~\cite{lu2018zebralancer,lu2020dragoon}. 
    We adopt a different approach; all proofs and responses are uploaded on-chain, but all verifications happen off-chain.
    Thus, to combat such behavior we can pair \avecq with a ``bounty-hunter'' protocol, where blockchain participants are incentivized to verify all requester-provided proofs and reap additional rewards upon discovering rejecting ones.

    We denote by $Adv^\mathcal{G}_x(\adv)$ the advantage that the adversary $\adv$ has in winning the game $\mathcal{G}_x$.
    Similarly we denote by $Adv^\textsf{C-Hiding}(\adv)$ the advantage $\adv$ has in breaking the hiding property of the commitment scheme \textsf{C}, 
    by $Adv^{H\textsf{-CR}}(\adv)$ breaking the collision-resistance property of $H$, 
    and by $Adv^\textsf{ProveQual-ZK}(\adv)$ breaking the zero-knowledge property of \textsc{ProveQual}.

\section{zk-SNARKs Specifications for $\mathsf{AVecQ}$}\label{sec:app-zk-snarks}
Figure~\ref{fig:zk-snarks} includes all the parameters and formal languages supported by the SNARKs used in $\mathsf{AVecQ}$.

\smallskip\noindent\underline{\textsc{ProveQual}}. For a task $\tau$, each worker $w_i$ generates a proof $\pi_{o_i}^\tau$ attesting to the correctness of its quality $q_i^\tau$ using \textsc{ProveQual}. The statement $\Vec{x}_i$ comprises the MT root, requester, and RA's public keys, the re-randomized commitment of $w_i$'s current quality, the quality ``tag'' and encryption of $w_i$'s answer and address for reimbursement. The corresponding witness $\Vec{\omega}_i$ consists of $w_i$'s certificate, its secret identifier, quality and corresponding randomness in plaintext, quality commitment, plaintext answer, and public address followed by the MT path. \textsc{ProveQual}'s language-specific checks prove the correctness of  $\Vec{x}_i$ while simultaneously guarding $w_i$'s response against free-rider attacks. As mentioned in Section~\ref{subsec:zk-snarks}, these checks include certificate verification $(\textsf{EdDSAver})$, quality verification $(\textsf{QualVer})$, correctness of the answer and address encryptions $(\textsf{ValidEnc})$, quality ``tag'' verification $(\textsf{HashComVer})$, MT membership proof $(\textsf{MTPathVer})$, and any task-specific check with $\textsf{TaskVer}$.

We can trivially see that these checks together ensure the validity of $\Vec{x}_i$ in \textsc{ProveQual}. We also remark that \textsc{ProveQual} assists in safeguarding \avecq against free-rider attacks. By including the correctness check for $w_i$'s reimbursement address in \textsc{ProveQual}, we ensure that an adversary copying $w_i$'s response cannot change the reimbursement address (without breaking the SNARK's soundness property). As such, no adversary has an incentive to launch free-riding attacks. Lastly, with Figure~\ref{fig:exp-provequal}, we also highlight the efficiency of \textsc{ProveQual}'s design. The proving time for the requester, with $>1$M responses (or depth $>20$), is 18-25 seconds.

Workers use \textsc{ProveQual} to generate proofs for their responses. In comparison, a task's requester uses \textsc{AuthCalc}, \textsc{AuthQual}, and \textsc{AuthValue} to attest to the correctness of the final answer, a worker's quality and its proximity to the final answer, respectively. 

\smallskip\noindent\underline{\textsc{AuthCalc}}. Each task's requester uses \textsc{AuthCalc} to generate the proof $\pi_{a_\phi}^\tau$ to attest the correctness of the final answer $a_\phi^\tau$. The statement $\Vec{x}_{AC}$ includes the encryption of $a_\phi^\tau$, the set of all encrypted responses and the public key $pk_R$; while the witness $\Vec{\omega}_{AC}$ comprises the requester's secret key $sk_R$. Each task's policy is hard-coded in the SNARK. The language makes up the following two checks: (i) \textsf{ValidKeyPair}, which ensures that $(pk_R,sk_R)$ are valid key pair and (ii) \textsf{FinAnsVer}, which checks if the final answer is correctly computed given the set of all encrypted responses. 

Figure~\ref{fig:exp-authcalc1} and Figure~\ref{fig:exp-authcalc2} shows the practicality of \textsc{AuthCalc}. E.g., the SNARK can generate proofs for popular crowdsourcing policies for $1$K workers and $64$ choices in $<15$ minutes.

\smallskip\noindent\underline{\textsc{AuthQual} and \textsc{AuthValue}}. The requester uses these SNARKs to generate the proofs $\pi_{q_i}^\tau$
and $\pi_{a_i}^\tau$. These attest to the correctness of $w_i$'s updated quality and the proximity of their answer with the final answer. Note that the quality update rule and worker answer policies are hard-coded in the SNARK. \textsf{AuthValue}'s statement $\Vec{x}_{AV}$ includes encryption of the final and $w_i$'s answer and requester public key $pk_R$, while \textsc{AuthQual}'s statement $\Vec{x}_{AQ}$ additionally includes the commitments of the old and updated qualities. Both the witnesses, $\Vec{\omega}_{AV}$ and $\Vec{\omega}_{AQ}$, comprise the secret key $sk_R$. The language for \textsc{AuthValue} comprises checks for the key-pair validity $(\textsf{ValidKeyPair})$ and the correctness of $w_i$'s answer with respect to $a_\phi^\tau$ $(\textsf{EqCheck})$. Likewise, \textsc{AuthQual}'s language includes \textsf{ValidKeyPair} and \textsf{NewQual} which checks the quality update given $w_i$'s answer, $a_\phi^\tau$ and the old and updated quality commitments as inputs. 

Notably, from Table~\ref{tab:2snarks}, both these SNARKs are efficient in proving and verification times for popular crowdsourcing policies. Depending on the policy, the proving time ranges from 2.9-3.6 seconds, while the verification takes $<12.5$ milliseconds.


\appendix

\end{document}